\documentclass[12pt]{article}
\usepackage{fullpage}
\usepackage{setspace}
\usepackage{graphicx}
\usepackage{amsmath, amssymb, amsthm}
\usepackage{natbib}
\usepackage{hyperref}

\newtheorem{mytheorem}{Theorem}

\newtheorem{myprop}{Proposition}
\newtheorem{mylemma}{Lemma}
\newtheorem{mydef}{Definition}

\newcommand{\vect}[1]{\mbox{\boldmath $ #1$}}

\title{
\vspace{-1cm}
Multi-resolution two-sample comparison through
the divide-merge Markov
tree}
\author{Jacopo Soriano and Li Ma  \\Department of
Statistical Science  \\ Duke University, Durham, NC 27708  \\ email:
\texttt{jacopo.soriano@duke.edu, li.ma@duke.edu} }

\begin{document}
\maketitle

\begin{abstract}
We introduce a probabilistic framework for two-sample comparison based
on a nonparametric process taking the form of a Markov model that
transitions between a ``divide'' and a ``merge'' state on a
multi-resolution partition tree of the sample space. Multi-scale
two-sample comparison is achieved through inferring the underlying
state of the process along the partition tree. The Markov design
allows the process to incorporate spatial clustering of differential
structures, which is commonly observed in two-sample problems but
ignored by existing methods. Inference is carried out under the
Bayesian paradigm through recursive propagation algorithms. We
demonstrate the work of our method through simulated data and a real
flow cytometry data set, and show that it substantially outperforms
other state-of-the-art two-sample tests in several settings.
\end{abstract}


\section{Introduction}

Two-group comparison is of fundamental interest in a variety of applications, 
and over the last few decades, many nonparametric two-sample tests have been
invented. Some notable examples include the
Kolmogorov-Smirnov test \cite[]{bickel_1969}, the $k$-nearest
neighbors test \citep{schilling_1986,henze_1988}, and the 
Cram\'{e}r test \cite[]{baringhaus_franz_2004}, among many others. 
More recently a number of Bayesian methods have been proposed, and 
in particular a multi-resolution approach based on the P\'olya tree (PT) process
\citep{ferguson_1973,lavine_1992} has gained popularity. The PT process
decomposes a probability distribution in a wavelet-like fashion into a
collection of local probability assignment coefficients on a multi-scale
partition sequence of the sample space. Two-sample comparison is then achieved
through comparing these local assignments and combining the statistical evidence
across the partition sequence to form a global test statistic. Methods following
this strategy include \cite{holmes_etal_2012}, \cite{chen_hanson_2012}, and
\cite{ma_wong_2011}. 

The main motivation for this work stems from an important and almost ubiquitous 
phenomenon in two-sample problems that has not been taken into account by
existing methods---the spatial clustering of differential structures. More
specifically, if there is two-sample difference at one location in the sample
space, then the neighboring parts of the sample space are more likely to contain
differential structures as well. Therefore methods for finding two-sample
differences across the sample space should ideally incorporate the dependence
among the neighboring locations rather than treating them independently. This
relates to the more general problem of multiple testing adjustment for dependent
hypotheses.

In this work we introduce a new method for two-sample comparison that
incorporates 
this dependence structure to improve power. The core of our method is a new
nonparametric process---called the {\em Divide-Merge Markov Tree}
(DMMT)---taking the form of a tree-structured Markov model
\citep{crouse_etal_1998} that transitions between ``divide'' and ``merge''
states on a multi-resolution partition sequence of the sample space,
corresponding respectively to whether the two distributions are locally equal or
not. The Markov dependence provides a convenient framework for incorporating
spatial dependence of differential structures. Under this framework, testing
two-sample differences is achieved through inferring the underlying Markov state
of the process on different parts of the space.

Moreover, in multivariate problems the DMMT allows the multi-resolution
partition 
sequence on which the Markov tree grows to be data-adaptive. 
When the sample space is vast, as is common in multi-dimensional problems,
fixed, 
symmetric partition sequences adopted by standard multi-resolution methods
quickly run into a sparsity problem---just a few levels down the partition
sequence, most sets contain only very few data points. Consequently, on most
location-scale combinations there are too few data to draw reliable inference.
This sparsity translates into high posterior uncertainty of the underlying
states for most sets in the partition sequence and loss of statistical power. 
To address this challenge, the DMMT adopts a random partitioning 
mechanism \citep{wong_ma_2010} that allows the partition sequence to be inferred
from the data. Interestingly this additional adaptiveness can be achieved
without
affecting the Markov nature of the process, which is critical for incorporating
spatial clustering and efficient posterior inference.

Furthermore, the DMMT process provides a natural means to summarizing and 
visualizing the inferred difference. 
We show that one can effectively represent posterior summary of the 
differential structure by plotting
a representative partition sequence of the sample space and highlighting 
the regions where the two distributions have high posterior probability of being
different.

In addition, we formally investigate the theoretical properties of the DMMT
model. 
We show that it satisfies Ferguson's criteria \citep{ferguson_1973} for
desirable nonparametric processes---it has large support (and so is a
nonparametric process) and is analytically tractable (posterior inference can be
carried out very efficiently due to its Markov nature). Moreover, we show that
the multi-resolution two-sample test using the DMMT is consistent. 
\cite{holmes_etal_2012} showed that the multi-resolution approach using 
the standard PT process gives consistent two-sample tests. Our results further
show that this consistency can be maintained while introducing additional
flexibility and adaptivity into the underlying model through the Markov design
and the adaptive partitioning feature.

The paper is organized as follows. 
In Section \ref{sec:dimeopt} we present our multi-resolution two-sample
comparison
 framework based on the DMMT process. We introduce the DMMT process, establish
its theoretical properties, and provide the inferential recipe for carrying out
two-sample comparison. We provide guidelines for prior specification.
In Section \ref{seq:examples} we illustrate the work of our method and evaluate
its performance. We first carry out a simulation study to compare the power of
our method as a two-sample test to other nonparametric two-sample tests. 
We also illustrate how to pin-point and visualize
where and what the difference is using the posterior process. 
Finally, we apply our method to analyzing a seven-dimensional cytometry data
set, 
for which our method successfully identifies an experimentally validated
differential hotspot involving just 0.2\% of the data points. 
 Section \ref{seq:conclusion} closes with a brief discussion.

\section{Method}\label{sec:dimeopt}

\subsection{Some basic concepts and notation on recursive partitioning}
Recursive partitioning of the sample space is a fundamental building block 
for the multi-resolution approach to two-sample comparison
\citep{holmes_etal_2012,chen_hanson_2012,ma_wong_2011}. We start by introducing
some basic concepts and notation about recursive partitioning that will be used
throughout the paper.

Let $\Omega$ denote the sample space, which can either be finite such as a 
contingency table or an Euclidean rectangle. While our proposed method is
applicable to both cases, for simplicity, our presentation will focus on the
Euclidean case. Without loss of generality, let $\Omega=[0,1)^p$. An unbounded
rectangle can be transformed into a bounded rectangle by applying, for example,
a cdf
transformation to each dimension. 

A {\em dimensionwise dyadic partition} of
$A=[a_1,b_1)\times[a_2,b_2)\times\cdots \times[a_p,b_p)\subset \Omega$ 
refers to a division of $A$ into two halves by splitting at the middle of the
support of one of the $p$ dimensions. For each $A$, we use $\{ A_l^j, A_r^j \}$
to denote
the pair of children nodes---called the \emph{left} and \emph{right}
children---of $A$ 
obtained by
cutting $A$ along the $j$th direction. That is, 
$A^j_l$ is the half with the $j$th dimension supported on $[a_j, (a_j+b_j)/2)$,
and $A^j_r$ is the half with the $j$th dimension supported on
$[(a_j+b_j)/2,b_j)$.

 We consider recursive partition sequences of $\Omega$ generated by
dimensionwise 
dyadic partitions. Each such partition sequence can be represented by a
bifurcating tree.
A node in the tree represents a subset of the sample space $\Omega$,
and is obtained from a dyadic partition of its parent node. 
At the top level of the tree, the root (or level-0) node of the tree is
 the whole space $\Omega$, while 
at the first level there are two (level-1) nodes obtained by partitioning $\Omega$ along
one of the $p$ dimensions. Each level-1 node can be partitioned again,
defining the next level of the tree and so on. 

Let $\mathcal{A}^k$ denote the collection of all level-$k$ nodes under 
all possible recursive partition sequences of $\Omega$. In other words, it is
the set of all possible subsets obtainable by sequentially partitioning the
sample space $k$ times. Also, let $\mathcal{A}^{(k)} = \cup_{i=0}^{k}
\mathcal{A}^i$, the collection of all possible nodes up to level $k$, and
$\mathcal{A}^{(\infty)} =  \cup_{i=0}^{\infty} \mathcal{A}^i$, the collection of
all possible nodes.

\subsection{The divide-merge Markov tree (DMMT)}
Next we introduce a generative model for a pair of distributions $(Q_1,Q_2)$, 
called the divide-merge Markov tree (DMMT). 
We first describe the general design principle of the process and then provide
 the mathematical details. 

The DMMT model takes a P\'olya tree like multi-resolution approach to
generating 
probability distributions. It divides the sample space into a sequence of nested
partitions, and then at each node in the partition tree, it specifies how
probability mass is split between the left and the right children. The process
uses this strategy to generate both distributions $Q_1$ and $Q_2$
simultaneously, and use a hidden Markov model \citep{crouse_etal_1998} to
determine the relationship between the two distributions on each of the node.
Under a ``divide'' state, probability mass is split differently for the two
distributions into the two children nodes, while under a ``merge'' state, the
same probability assignment is applied for both distributions. The Markov tree
also has an additional ``stop'' state, representing the case when the two
distributions are conditionally equal to some baseline distribution. 

Under this design, inference on two-sample comparison can be achieved through 
learning the divide-merge states of the process across the partition tree. To
further allow the partition sequence of the sample space to be data-adaptive
{\em a posteriori} thereby improving statistical power in multivariate problems,
the DMMT incorporates a randomized partition mechanism introduced in
\cite{wong_ma_2010}. 

Given this general design, next we formally describe the DMMT process as a 
generative procedure in an inductive manner. Suppose the generative procedure
has proceeded onto a subset $A$ of $\Omega$. To initiate the induction,
$A=\Omega$, the entire sample space. The procedure carries out three operational
steps on $A$:
\begin{enumerate}
 \item Divide-merge step---to determine its ``divide'', ``merge'', or
``stop'' state on $A$;
 \item Random partitioning step---to determine how to further partition $A$ into
children;
 \item Probability assignment step---to determine how $Q_1$ and $Q_2$ assign probability to $A$'s children.
\end{enumerate}
Next we describe each of the steps in turn. Figure \ref{fig:3steps} 
provides an illustration to help the reader understand the technical
description.
\begin{figure}
 \centering
 \makebox{\includegraphics[width=0.95\textwidth]{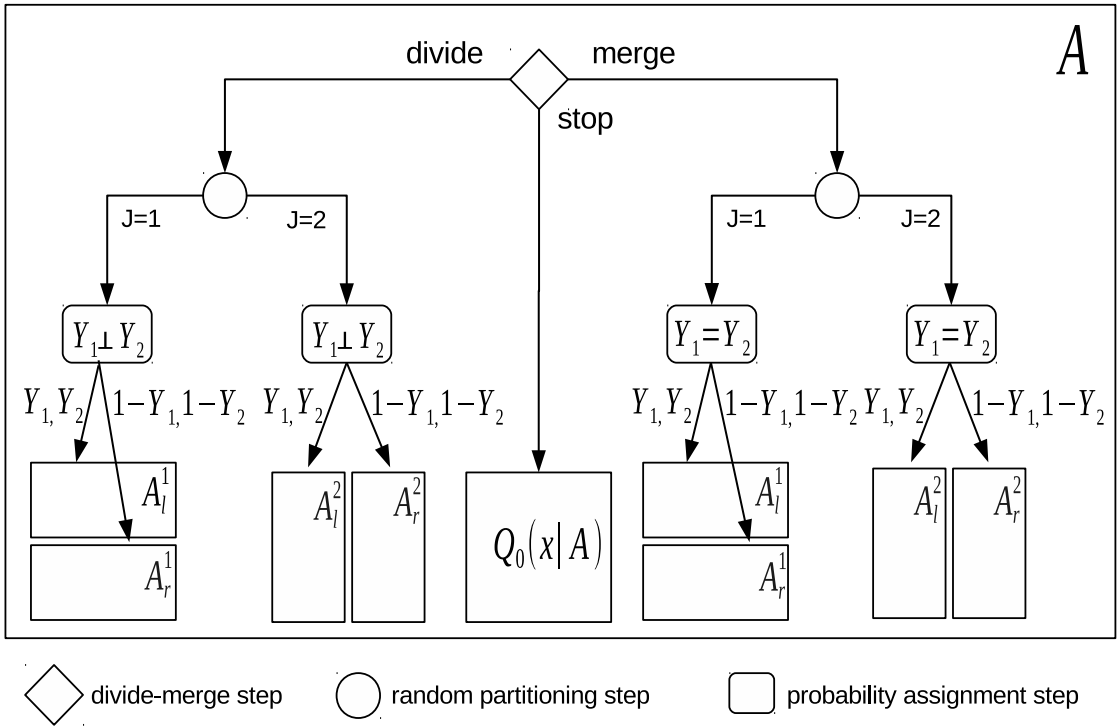}}
 \caption{The three steps of the DMMT generative procedure for
a set $A \in \mathcal{A}^{(\infty)}$.}
\label{fig:3steps}
\end{figure}

\paragraph{Divide-merge step}
The procedure can take one of three states: ``divide'', ``merge''
and ``stop'' on $A$. 
Define the set of the  states $\mathcal{G} =
\{d,m,s\}$. 
Consider $h \in \mathcal{G}$ being the
state of the parent of $A$. 
We draw a random variable to determine the state of the process on $A$:
$$
\Pr \big{[} S(A) = g |  S(\text{parent}(A))=h   \big{]} =
\rho_{h,g}(A),
$$
where $\rho_{h,g}(A)\geq0$ for any $g,h \in \mathcal{G}$, and
$\sum_{g}\rho_{h,g}(A) = 1$ for any $h \in \mathcal{G}$.  We call
$\rho_{h,g}(A)$ the transition probabilities, and organize them in a
transition probability matrix
$$
\vect{\rho}(A) = 
\left[
\begin{array}{ccc}
 \rho_{d,d}(A) & \rho_{d,m}(A) & \rho_{d,s}(A) \\
 \rho_{m,d}(A) & \rho_{m,m}(A) & \rho_{m,s}(A) \\
 \rho_{s,d}(A) & \rho_{s,m}(A) & \rho_{s,s}(A) 
\end{array}
\right].
$$
The procedure will terminate once in the ``stop'' state,
and so by construction $\rho_{s,d}(A)=\rho_{s,m}(A)= 0$ and
$\rho_{s,s}(A) = 1$. The root of the tree
$\Omega$ does not have a parent, so we draw its ``parent'' state from a
multinomial 
Bernoulli corresponding to ``initial'' state probabilities  $ (\rho_{0,d},
\rho_{0,m},\rho_{0,s} )$.

\paragraph{Random partitioning step}
If the process ``stops'' on $A$, i.e.\ $S(A)=s$, then $A$ is not further
partitioned. 
Otherwise, if $S(A) = g$ for $g \in  \{ d,m \}$, we draw a
partition direction $J(A)$ according to
$$
\Pr\big{(}J(A) = j| S(A)=g \big{)} = \lambda_{j}(A,g),
$$
where $\lambda_j(A,g)\geq 0 $ for  $j=1, \ldots, p$ and $\sum_j
\lambda_j(A,g)=1$ for 
$g \in
 \{ d,m \}$, and we make a dimensionwise dyadic partition in dimension $j$ and obtain two
subsets $A^{j}_l$
and
$A^{j}_r$. 
It may look intriguing that the partition probability $\lambda_j(A,g)$ should
be allowed to depend on the state $g$. We will see later that under $g=m$, the
partition characterizes the common structure of the two distributions, while
under $g=d$, their difference.

\paragraph{Probability assignment step} The probability assignment depends on the
state of the process.
If $S(A) =d $, we generate two independent Beta variables $Y_1(A)$ and $Y_2(A)$
such that
$$
Y_t(A)  \sim \text{Beta}\big{(}
\vect{\alpha}^j_t(A,d) \big{)},
$$
where $ \vect{\alpha}^j_t(A,d) = \big{(} \alpha_{t,l}^j(A,d),
\alpha_{t,r}^j(A,d) \big{)}$ are the pseudo-counts parameters, 
which can be different for $t=1$ and 2.
Otherwise, if $S(A)=m$, we generate 
$$
Y_1(A) = Y_2(A) \sim \text{Beta}\big{(}
\vect{\alpha}^j(A,m) \big{)},
$$
where $ \vect{\alpha}^j(A,m) = \big{(} \alpha_{l}^j(A,m),
\alpha_{r}^j(A,m) \big{)}  $ are the pseudo-counts parameters.
For $t=1$ and 2, we assign a proportion equal to $Y_{t}(A)$ and
$1-Y_{t}(A)$ of the probability that $Q_t$ assigns to $A$ to the left and right children,
respectively:
$$
 Q_t(A^j_l)  = Q_t(A) Y_{t}(A), \quad Q_t(A^j_r)  = Q_t(A)(1- Y_{t}(A)).  
$$
Finally, if $S(A)=s$, then for both $Q_1$ and $Q_2$, probability is assigned 
on $A$ accordingly to a
baseline  distribution $Q_0$
$$ 
Q_1(\cdot|A) = Q_2(\cdot|A) = Q_0(\cdot|A).
$$

This completes the inductive description of the procedure on $A$. 
If the procedure does not stop on $A$, then it proceeds onto applying the same
three-step operation on each of $A$'s children. The generative model is
completely specified by the three sets of parameters
$\vect{\rho}(A)$, $\lambda_j(A,g)$ and $\vect{\alpha}^j_t(A,g)$ for $j=1,\
\ldots, p$, $t=1,2$, $g \in \{ d,m\}$ and $A\in \mathcal{A}^{(\infty)}$.
From now on we will write these three sets of parameters as
$\vect{\rho}$, $\vect{\lambda}$ and $\vect{\alpha}$ respectively.

If the prior probability of stopping is uniformly bounded away from zero, 
i.e.\ there exists $\delta>0$ such that 
$\rho_{g,s}(A)>\delta$ for all $A \in \mathcal{A}^{(\infty)}$ 
and $g\in \mathcal{G}$, then with probability~1 the generative procedure will
stop almost everywhere on the
space $\Omega$.
Under the same condition, using a similar argument as in Theorem~2 in
\cite{wong_ma_2010}, 
one can show that with probability~1 this procedure will produce a pair of
well-defined probability measures $(Q_1,Q_2)$ that are both absolutely
continuous 
with respect to the baseline measure $Q_0$. 
Thus, we can formally define this generative model as a distribution on a pair 
of probability measures.
\begin{mydef}
 The pair of random probability measures $(Q_1, Q_2)$ is said to have a
divide-merge Markov tree (DMMT) distribution with parameters
$\vect{\rho}, \vect{\alpha}, \vect{\lambda}$,and baseline measure $Q_0$. 
We write $(Q_1, Q_2) \sim
\text{DMMT}(\vect{\rho},
\vect{\alpha}, \vect{\lambda}, Q_0)$. 
\end{mydef}

The next theorem shows that the DMMT has large $L_1$ support, and thus can
 be used as a nonparametric prior. 
\begin{mytheorem}\label{thm:large_support}
Assume $(Q_1,Q_2)\sim
\text{DMMT}(\vect{\rho},
\vect{\alpha}, \vect{\lambda}, Q_0)$ where
\begin{enumerate}
 \item The initial state $\rho_{0,s}<1$.
 \item The transition
probabilities $\rho_{d,d}(A), \rho_{d,m}(A), \rho_{d,s}(A), \rho_{m,m}(A),
\rho_{m,s}(A)$, the
direction probabilities $\lambda_j(A,g)$ and
the pseudo-count parameters $\alpha_t(A_l^j,g)/(\alpha_t(A_l^j,g) +
\alpha_t(A_r^j,g) )$ are uniformly bounded away from 0 and 1, for all $A \in
\mathcal{A}^{(\infty)}$ and $g \in \{d,m\}$.
\end{enumerate}
Then, for any pair of distributions $F_1$ and $F_2$ absolutely continuous w.r.t.
$Q_0$ and any $\tau>0$, we have
$$
\Pr\big{(}\int |\tilde{q}_t-\tilde{f}_t|dQ_0<\tau, \text{ for }t=1,2 \big{)} >
0,
$$ 
where $\tilde{q}_t=dQ_t/dQ_0$ and $\tilde{f}_t=dF_t/dQ_0$.   
\end{mytheorem}
\begin{proof}
See Supplementary Material.
\end{proof}

One can center a DMMT distribution at the baseline measure $Q_0$. More specifically, if
$\text{DMMT}(\vect{\rho},
\vect{\alpha}, \vect{\lambda}, Q_0)$ satisfies:
\begin{enumerate}
 \item There exists a $\delta>0$ such that $\rho_{g,s}(A)>\delta$ for all $A \in
\mathcal{A}^{(\infty)}$ and $g\in \mathcal{G}$;
 \item For all $A \in \mathcal{A}^{(\infty)} $ such that $Q_0(A)>0$, $j=1,
\ldots p$ and $g\in \{d,m\}$:
\begin{equation}\label{eq:centering_condition}
\dfrac{\alpha_{t,l}^j(A,g)}{ \alpha_{t,l}^j(A,g) + \alpha_{t,r}^j(A,g)} =
\dfrac{Q_0(A^j_l)}{Q_0(A)},
\end{equation}
\end{enumerate}
then
$ E ( Q_t(B) ) = Q_0(B)$  for any
$ B \in \mathcal{B}(\Omega) $ 
and $E(q_t(x)) = q_0(x)$, where $t=1,2$, $q_t = dQ_t/d\mu$ and $q_0 = dQ_0/d\mu$
with $\mu$ being a 
dominating measure such as the Lebesgue measure. We shall refer to condition
\eqref{eq:centering_condition} as the centering condition.

In some situations, we have a general idea of the common shape of the
two distributions, and are able to elicit it with a simple distribution from a
given parametric family. 
In this case, one can center the DMMT prior at a given $Q_0$ that reflects 
this prior knowledge.
This helps to achieve parsimony, allowing the process to focus on modeling the
differential
structure between the two distributions, and thus yielding increased power 
in two-sample comparison.

\subsection{The posterior of a DMMT prior}\label{sec:bayes_inference}

In this section we find the corresponding posterior of a DMMT prior. 
Later we will use the posterior to carry out inferential tasks such as
two-sample comparison in a Bayesian paradigm. Our main result
(Theorem~\ref{thm:conjugacy}) shows that the DMMT process is posterior
conjugate, in the sense that if $(Q_1,Q_2)$ has a DMMT prior, then after
observing i.i.d.\ data samples from $Q_1$ and $Q_2$, the posterior is still a
DMMT distribution. Moreover, the posterior parameters can be computed
recursively. This intriguing result follows from the Markov nature of the DMMT
process, and the results in this section correspond to a
``forward-summation-backward-sampling'' algorithm for computing posterior Markov
models \citep{liu_2008}. More specifically, Lemma~\ref{lemma:recursive_phi}
corresponds to the ``forward-summation'' step for the DMMT process while
Theorem~\ref{thm:conjugacy} the ``backward-sampling'' step.
Lemma~\ref{th:terminal_nodes} shows that the recursion can be carried out
analytically even when the partition sequence is infinite. Readers less
interested in the technical details may directly jump to those results. 

Assume that
$(Q_1, Q_2)\sim
\text{DMMT}(\vect{\rho},
\vect{\alpha}, \vect{\lambda}, Q_0)$, and we observe two groups of 
i.i.d.\ samples
$\vect{x}_1=(x_{1,1}, \ldots,
x_{1,n_1})$ and
$\vect{x}_2=(x_{2,1}, \ldots, x_{2,n_2})$ on $\Omega$  from
$Q_{1}$ and
$Q_{2}$,
respectively.  
In the following, we shall use $\vect{x}_1 \cup \vect{x}_2$ to represent the
pooled sample
$(x_{1,1}, \ldots,x_{1,n_1},x_{2,1}, \ldots,x_{2,n_2})$.
Let $q_t(x|A)= q_t(x)/Q_t(A)$ for $x\in A$  be the conditional density on
$A$ of $Q_t$  for  $t=0,1,2$. 
Then,  the conditional likelihood on $A$ is 
\begin{equation}\label{eq:likelihood}
  q(\vect{x}_1, \vect{x}_2|A) = \prod_{t=1,2} q_t(\vect{x}_t|A),
\end{equation}
where
$q_t(\vect{x}_t|A) = \prod_{x_{t,j} \in A} q_t(x_{t,j}|A)$ for $t=1,2$. 
We 
introduce a mapping
$\Phi: \mathcal{A}^{(\infty)} \times \mathcal{G} \times \Omega^{n_1} \times
\Omega^{n_2} \mapsto \mathbb{R}$,  which will be useful in expressing the
posterior of a DMMT,
\begin{equation}\label{eq:marginal_likelihood}
\Phi(A,g,\vect{x}_1,\vect{x}_2) := \int q( \vect{x}_1,\vect{x}_2|A)
\pi(dQ_1, dQ_2|E_1(A), E_{2,g}(A) ),
\end{equation}
where 
$\pi(\cdot, \cdot)$ denotes the prior of $(Q_1,Q_2)$,
$E_1(A) = \{ A \text{ arises during the random partitioning}\}$ and
$E_{2,g}(A) = \{ S(\text{parent}(A))=g \} $ and $g \in \mathcal{G}$.
This quantity represents the conditional likelihood on $A$ integrated with
respect to the prior on $(Q_1,Q_2)$, given that the parent of $A$ is in 
state $g$ and $A$ arises during the random partitioning. Notice that the
DMMT process restricted on $A$ is still a
DMMT on $A$ due to its self-similarity, thus
\eqref{eq:marginal_likelihood} represents the marginal likelihood for a DMMT
with root being $A$ and initial state $\rho_{0,g}=1$.
Lemma \ref{lemma:recursive_phi} provides a recursive representation of 
\eqref{eq:marginal_likelihood}. 
\begin{mylemma}\label{lemma:recursive_phi}
For every $A \in \mathcal{A}^{(\infty)}$ and $g\in \mathcal{G}$,
$\Phi(A,g,\vect{x}_1,\vect{x}_2)$ has the following recursive representation
\begin{equation*}
\Phi(A,g,\vect{x}_1,\vect{x}_2) = \sum_{h \in \mathcal{G}} \rho_{g,h}(A)
Z(A,h,\vect{x}_1,\vect{x}_2),
\end{equation*}
where
\begin{equation*}
 \begin{split}
  Z(A,g,\vect{x}_1,\vect{x}_2) & =
\left\{ \begin{array}{ll}
			    \sum_{j=1}^{p} Z_j(A,g,\vect{x}_1,\vect{x}_2) &
\text{if }
g\in  \{d,m\} \\
		\prod\limits_{x \in \vect{x}_1 \cup \vect{x}_2 } q_0(x|A) &
\text{if }
g = s,
                           \end{array}
                   \right. \\
Z_j(A,g,\vect{x}_1,\vect{x}_2) & =  \left\{ \begin{array}{ll}
\lambda_j(A,m) \dfrac{D( \vect{\alpha}^j(A,m) +
\vect{n}_1^j(A) + \vect{n}_2^j(A)
)}{D( \vect{\alpha}^j(A,m) )}    \prod\limits_{i \in \{l,r\}}
\Phi(A^j_i,m,
\vect{x}_1,\vect{x}_2) 
& \text{ if } g=m \\
\lambda_j(A,d) \prod\limits_{t=1,2}\dfrac{D(
\vect{\alpha}^j_t(A,d) + \vect{n}^j_t(A) )}{D( \vect{\alpha}^j_t(A,d)  )} 
\prod\limits_{i \in \{l,r\}}
\Phi(A^j_i,d, \vect{x}_1,\vect{x}_2) 
& \text{ if } g=d,\\
			    \end{array}
                   \right.
 \end{split}
\end{equation*}
for $j=1, \ldots, p $,  $t=1,2,$ $\vect{n}^j_t(A) = \big{(} n_t(A_l^j),
n_t(A_r^j)\big{)}$, $n_t(A) = | \{x_{t,i} : x_{t,i} \in A,\; i=1,2, \ldots,
n_t\} |$ and 
$D(w_1,w_2) = \Gamma(w_1)\Gamma(w_2)/ \Gamma(w_1 + w_2)$.
\end{mylemma}
\begin{proof}
 See Supplementary Material.
\end{proof}

This representation of $\Phi$ is recursive in its first argument in the sense
that one can compute $\Phi(A,\cdot,\cdot,\cdot)$ based on 
$\Phi(A^j_i,\cdot,\cdot,\cdot)$. 
This recursive representation becomes
operational if 
\eqref{eq:marginal_likelihood} can be eventually expressed in closed form.
 Lemma \ref{th:terminal_nodes} provides analytic expressions for some specific
regions.
\begin{mylemma}\label{th:terminal_nodes}
For two types of  
 regions   $\Phi(A,g,\vect{x}_1,\vect{x}_2)$ is known analytically:
 \begin{enumerate}
 \item Empty regions, i.e. $ A : (\vect{x}_1 \cup\vect{x}_2) \cap A = \emptyset
$, then
$\Phi(A,g,\vect{x}_1,\vect{x}_2)
= 1$;
 \item Regions with a single observation, i.e. $ A: |(\vect{x}_1 \cup\vect{x}_2)
\cap A|=1$, then
$\Phi(A,g,\vect{x}_1,\vect{x}_2) = \prod_{x \in \vect{x}_1 \cup \vect{x}_2}
q_0(x|A)$ under the centering condition.  
\end{enumerate}
\end{mylemma}
\begin{proof}
 See Supplementary Material.
\end{proof}

For every finite sample size $n_1 + n_2$, there is a finite partitioning level
$k$ such that all the 
nodes of  level $k$ in the partitioning tree  belong to
one of these two types of ``terminal'' nodes. Thus,
\eqref{eq:marginal_likelihood} can
be computed recursively
from these nodes of
the tree up to the root. 
Finally, Theorem \ref{thm:conjugacy} establishes the conjugacy of DMMT and
provides expression for the posterior parameters. 
\begin{mytheorem}\label{thm:conjugacy}
Suppose we observe two groups of i.i.d.\ samples $\vect{x_1} = (x_{1,1}, \ldots,
x_{1,n_1})$ and
$\vect{x_2} = (x_{2,1}, \ldots, x_{2,n_2})$ from two distributions $Q_1$ and
$Q_2$. If $(Q_1,Q_2)$  have a
DMMT$(\vect{\rho},\vect{\lambda},\vect{\alpha}, Q_0)$
prior, then, the posterior of
$(Q_{1},Q_{2})$ is still a
DMMT with the same baseline $Q_0$ and the
following
parameters:
\begin{enumerate}
 \item Transition probabilities: 
$$
\rho_{g,h}(A|\vect{x}_1,\vect{x}_2) = \rho_{g,h}(A) \dfrac{
Z(A,h,\vect{x}_1,\vect{x}_2)}{\Phi(A,g,\vect{x}_1,\vect{x}_2)}
\quad \text{for all } A \in \mathcal{A}^{(\infty)}, g,h\in\mathcal{G}. $$
 \item Direction probabilities:
$$
\lambda_j(A,g|\vect{x}_1,\vect{x}_2) = \dfrac{ Z_j(A,g,\vect{x}_1,\vect{x}_2) }{
Z(A,g,\vect{x}_1,\vect{x}_2) } \quad \text{for all } A \in
\mathcal{A}^{(\infty)}, g\in \{ d,m\} \text{ and } j=1,\ldots,p.
$$
 \item Pseudo-counts:
$$
\alpha_t(A,g|\vect{x}_1,\vect{x}_2) = \left\{ \begin{array}{ll}
            \alpha_t(A,d)  + n_t(A) & \text{if } g= d \\
            \alpha_t(A,m)  + n_1(A) + n_2(A) & \text{if } g= m,
                           \end{array}
\right. 
$$
for all $A \in \mathcal{A}^{(\infty)}$ and $t=1,2$. 
\end{enumerate}
\end{mytheorem}
\begin{proof}
 See Supplementary Material.
\end{proof}

\subsection{ Two-sample testing}\label{subsec:2sampleTesting}

Two probability measures $(Q_1,Q_2)$ with a DMMT distribution are identical if
and only
if the DMMT process,
in any part of the partition tree, is never in the divide state. 
The posterior probability of this event, $H_0$, is given by
\begin{equation}\label{eq:same_dist}
 \Pr(H_0|\vect{x}_1,\vect{x}_2) =  \int 1\left(Q_1(\cdot)=Q_2(\cdot)\right)
\pi(dQ_1, dQ_2|\vect{x}_1,\vect{x}_2),
\end{equation}
where $\pi(\cdot,\cdot|\vect{x}_1,\vect{x}_2)$ denotes the posterior DMMT.
This probability can be estimated by drawing samples
of the two distributions from the posterior, and computing the proportion of
draws that are never in the divide state, i.e.,
$$
\begin{array}{cc}
\dfrac{1}{B} \sum_{b=1}^B 1\left(Q_1^{(b)}(\cdot)=Q_2^{(b)}(\cdot)\right), & \text{where }
\left(Q_1^{(b)},Q_2^{(b)}\right) \sim (Q_1, Q_2)|\vect{x}_1,\vect{x}_2. 
\end{array}
$$ 
Alternatively, again due to the Markov nature of the DMMT, one can 
evaluate $\Pr(H_0| \vect{x}_1,\vect{x}_2)$ analytically through another
``forward-summation'' type recursion, eliminating the Monte Carlo error involved
in posterior sampling.
To see this, we define another mapping $\Psi : \mathcal{A}^{(\infty)} \times 
\mathcal{G} \times \Omega^{n_1} \times \Omega^{n_2} \mapsto [0,1]$.
$$
\Psi(A,g,\vect{x}_1,\vect{x}_2)= \int 1\left(Q_1(\cdot|A)=Q_2(\cdot|A)\right)
\pi(dQ_1, dQ_2|\vect{x}_1,\vect{x}_2, E_1(A),E_{2,g}(A)).
$$
This function represents the marginal posterior probability that the two
distributions are identical conditional on $A$, given that
$A$ arises as a node in the partition tree and
the process is in state $g$ on the parent of $A$.

\begin{mylemma}
 The mapping $\Psi(A,g,\vect{x}_1,\vect{x}_2)$ has the following recursive
representation
\begin{multline}\label{eq:psi_recursive}
\Psi(A,g,\vect{x}_1,\vect{x}_2) =
\rho_{g,s}(A|\vect{x}_1,\vect{x}_2)  \\
 +
\rho_{g,m}(A|\vect{x}_1,\vect{x}_2)\displaystyle\sum_{j=1}^{p}
\lambda_j(A,m|\vect{x}_1,\vect{x}_2) \displaystyle\prod_{i \in \{l,r\}}
\Psi(A^j_i,m,\vect{x}_1,\vect{x}_2), 
\end{multline}
for all $A \in \mathcal{A}^{(\infty)}$ and $g \in \mathcal{G}$.
\end{mylemma}
\begin{proof}
 See Supplementary Material.
\end{proof}
For  $g=s$, 
$\Psi(A,g,\vect{x}_1,\vect{x}_2) = 1$, since
$\rho_{g,s}(A|\vect{x}_1,\vect{x}_2) = \rho_{g,s}(A) = 1$. For all other
regions and states, it can be evaluated with arbitrary
precision by terminating the recursive at a deep enough  finite level 
because $0 \leq \Psi(A,g,\vect{x}_1,\vect{x}_2) \leq 1$ for all $A \in
\mathcal{A}^{(\infty)}$ and $g \in \mathcal{G}$.
In particular, we can compute $ \Psi(\Omega,g,\vect{x}_1,\vect{x}_2)$ with
arbitrary precision.
Now we can formally describe how the posterior probability that the two
distributions are identical can be computed.

\begin{myprop}
Suppose we observe two groups of i.i.d.\ samples $\vect{x}_1=(x_{1,1}, \ldots,
x_{1,n_1})$ and
$\vect{x}_2=(x_{2,1}, \ldots, x_{2,n_2})$ from two distributions $Q_{1}$ and
$Q_{2}$.
Let $(Q_{1},Q_{2})$ have a DMMT($\vect{\rho},
\vect{\lambda}, \vect{\alpha}, Q_0$)
prior. Then the
posterior probability that the two distributions are
identical is given by
$$
\Pr(H_0|\vect{x}_1,\vect{x}_2) =  \sum_{g \in \mathcal{G}}
\rho_{0,g}\Psi(\Omega,g,\vect{x}_1,\vect{x}_2).
$$ 
\end{myprop}
\begin{proof}
 See Supplementary Material.
\end{proof}

The next two theorems establish the consistency for the 
two-sample test using DMMT.
\begin{mytheorem}
\emph{(Consistency under the alternative)}
\label{th:altcase}
We observe two independent groups of i.i.d.\ samples $\vect{x_1} = (x_{1,1},
\ldots, x_{1,n_1})$ and
$\vect{x_2} = (x_{2,1}, \ldots, x_{2,n_2})$ from two distributions $Q_1$ and
$Q_2$, where $n = n_1+n_2
\rightarrow \infty$, and  $n_1/n\rightarrow
\beta$ for some $\beta \in (0,1)$. 
Let
$(Q_1,Q_2)$ have a
DMMT$(\vect{\rho},\vect{\lambda},\vect{\alpha},
Q_0)$ prior where the conditions of Theorem \ref{thm:large_support} and the
centering condition are satisfied. In addition, if 
 $\rho_{d,g}(A) \in (0,1)$ for all
$g \in
\mathcal{G}$  and $A \in \mathcal{A}^{(k)}$ for some large enough
$k$,
then, 
$$
\Pr( H_0 | \vect{x}_1, \vect{x}_2) \overset{p}{\to} 0 \text{ under }
P_1^{( \infty)} \times P_2^{( \infty)},
$$
 for $P_1\ll Q_0$ and $P_2 \ll Q_0$ provided that there exists $A
\in \mathcal{A}^{(\infty)}$ such that $P_1(A_i^j|A) \neq
P_2(A_i^j|A)$ and $\beta P_1(A_i^j|A) + (1-\beta) P_2(A_i^j|A) \neq
Q_0(A_i^j|A)$ for some $j=1, \ldots, p$ and $i=l,r$.
\end{mytheorem}
\begin{proof}
 See the Supplementary Material.
\end{proof}

\begin{mytheorem}
\emph{(Consistency under the null)}
\label{th:null_case}
We observe two independent groups of i.i.d.\ samples $\vect{x_1} = (x_{1,1},
\ldots, x_{1,n_1})$ and
$\vect{x_2} = (x_{2,1}, \ldots, x_{2,n_2})$ from two distributions $Q_1$ and
$Q_2$, where $n = n_1+n_2
\rightarrow \infty$, and  $n_1/n\rightarrow
\beta$ for some $\beta \in (0,1)$. 
Let
$(Q_1,Q_2)$ have a
DMMT$(\vect{\rho},\vect{\lambda},\vect{\alpha},
Q_0)$ prior where   the
centering condition are
satisfied. In addition, if for some $k \in \mathbb{N}$
\begin{enumerate}
 \item $\rho_{g,h}(A) \in (0,1)$ for all
 $g \in
\mathcal{G}$, $h\in\{d,m\}$ and $A \in \mathcal{A}^{(k)}$.
 \item $\rho_{g,d}(A) = 0 $ for all $g\in\mathcal{G}$, $A\in \mathcal{A}^l$ and
for all $l>k$,
\end{enumerate}
then, 
$$
\Pr( H_0| \vect{x}_1, \vect{x}_2) \overset{p}{\to} 1 \text{ under }
P_0^{(\infty)} \times P_0^{(\infty)} \text{ for any } P_0 \ll Q_0.
$$

\end{mytheorem}
\begin{proof}
 See the Supplementary Material.
\end{proof}

\subsection{ Identifying the differential structure}
\label{subsec:identify_difference}

In many applications it is important not only to test two-sample difference, but
to
identify the differential structure. The multi-resolution approach provides 
a natural means to identifying differences through specifying the location,
scale,
and effect size of the difference.
All such information is encapsulated in the full DMMT posterior  and
can be extracted through posterior Monte Carlo sampling or constructing 
summaries of the posterior.
Here we design a strategy for effectively
summarizing such posterior information. This is particularly useful when 
one is interested in
learning the general structure of the difference rather than checking 
specific aspects of the difference that one has in mind. (The latter task is
more readily achieved through Monte Carlo sampling.) 

Our strategy is first identify a  representative partition tree that 
effectively characterizes the two-sample difference and
then report all regions in this tree that are {\em a posteriori} very 
likely to be in the divide state. 
Our representative partition tree is determined using two types of
posterior summaries that respectively shed light on how to further partition 
a tree node and when to stop the partitioning. More specifically, we summarize
the posterior information
regarding how to further partition a tree node using the posterior direction
probabilities weighted accordingly to the
marginal posterior probability of each non-stopping state
$$
\lambda_j^*(A|\vect{x}_1, \vect{x}_2) = \sum_{g\in\{ d,m\}}
\rho^*_{g}(A|\vect{x}_1, \vect{x}_2) \lambda_{j}(A,g|\vect{x}_1, \vect{x}_2),
$$
where $A \in \mathcal{A}^{(\infty)}$, $j=1, \ldots, p$. 
The summary for determining whether to stop further partitioning is 
\begin{equation}\label{eq:recursive_rho_star}
\begin{array}{ll}
  \rho^*_{g}(A|\vect{x}_1, \vect{x}_2) = P\big{(} S(A)=g | \vect{x}_1,
\vect{x}_2, B(A)\big{)}  &   \text{ for } g\in \mathcal{G},
\end{array}
\end{equation}
where  $B(A)$ is the collection of ancestral sets of $A$ in the partition tree, 
i.e.\ the branch of the partition tree that extends from $\Omega$ to $A$.
This probability is marginalized over the hidden states along this branch.
It can be computed recursively:
$$
\rho^*_{g}(A|\vect{x}_1, \vect{x}_2) = \left\{
\begin{array}{ll}
\sum_{h \in \mathcal{G} } \rho_{0,h}
\rho_{h,g}(A|\vect{x}_1, \vect{x}_2)   & \text{ if } A = \Omega \\
 \sum_{h
\in \mathcal{G}}\rho^*_{h}( \text{parent}(A)|\vect{x}_1, \vect{x}_2)
 \rho_{h,g}(A|\vect{x}_1, \vect{x}_2)  &  \text{ if } A \in
\mathcal{A}^{(\infty)}\setminus\{\Omega\},
\end{array}
\right.
$$
for all $g \in \mathcal{G}$.

Based on these two types of summaries, the representative tree is
constructed using the following top-down sequential procedure. 
Starting from the root $A=\Omega$,
if $ \rho^*_s(A|\vect{x}_1, \vect{x}_2)  >  1-\delta^*$ for some
threshold $\delta^* \in (0,1)$, then we stop; otherwise we divide the tree in
the direction $j$ that maximizes $\lambda_j^*(A|\vect{x}_1, \vect{x}_2)$.
For each $A_i^j$ we repeat this procedure until all branches are stopped.

Additionally, we define a notion of {\em effect size} to summarize the extent 
to which the two distributions are different on each reported region.
 The {\em effect size} on $A$  is given by
\begin{equation}\label{eq:effect_size}
 \text{Eff}(A|\vect{x}_1, \vect{x}_2) = \max_{j=1, \ldots, p} \bigg|
\log\dfrac{ \alpha_1(A_l^j,d|\vect{x}_1,
\vect{x}_2)/\alpha_1(A_r^j,d|\vect{x}_1, \vect{x}_2)
}{\alpha_2(A_l^j,d|\vect{x}_1, \vect{x}_2)/\alpha_2(A_r^j,d|\vect{x}_2,
\vect{x}_2)} \bigg|.
\end{equation}
This quantity is large when the proportion of probability mass in
the two children regions $A_l^j, A_r^j$ is very different between the two
groups for at least one of the $p$ dimensions.

\subsection{Prior specification and spatial clustering of
differences}\label{sec:prior}

The number of parameters characterizing the DMMT process is infinite, but
structural and symmetric assumptions on the prior parameters can
substantially simplify the specification of the prior.
In this  section we suggest default choices for these parameters. 

A simple and reasonable choice for  the
prior transition matrix $\vect{\rho}(A)$ is to make it
dependent on the level of the $A$, .i.e.,
$ \vect{\rho}(A) = \vect{\rho}(k)$.
Additionally, we adopt a
parsimonious and yet flexible functional form to further reduce the elicitation 
of the prior to a small
number of parameters. More specifically, we define 
$$
\vect{\rho}(k) = 
\left[
\begin{array}{ccc}
 \beta   &  (1-\beta) / 2&  (1-\beta) /2 \\
 \gamma 2^{-k}  &  (1-\gamma 2^{-k}) / 2&  (1-\gamma 2^{-k}) / 2\\
 0 & 0 & 1 
\end{array}
\right]
,
$$
for some $\beta, \gamma \in (0,1)$. 
The parameter $\beta$
determines the spatial clustering of the differential structure.
Given a set on which the process is in the divide state, 
a larger $\beta$ will suggest a higher likelihood of children and neighbors
also being in the divide state.
The parameter $\gamma$ controls 
  the transition from the merge state to
the divide state, and the $2^{-k}$ factor is included to provide adequate
control for multiple testing.
The number of sets in $\mathcal{A}^k$
increases geometrically in the depth $k$, and the $2^{-k}$ factor ensures
  that the prior probabilities of the
divide state decrease accordingly. 

The parameters $\beta$ and $\gamma$ along with the initial
state probabilities  controls the prior
marginal null probability that the two distributions are identical. 
 Because the stop state is irreversible, one should fix $\rho_{0,s}=0$.
This leaves three free parameters 
$\beta$, $\gamma$, and $\rho_{0,d}$ to be specified.
A useful strategy to choosing these parameters is to set the prior marginal null
probability to some moderate value such as 0.5. Simple choices such
as
$\beta=0.3$, $\gamma=0.2$ and $\rho_{0,d}=1$ 
that satisfy this constraint are highly
effective in a variety of numerical examples.
A sensitivity analysis  shows that the posterior inference
is
robust to the
prior specification. The results of the analysis are reported in 
the Supplementary Material.

Without prior knowledge about what dimensions are more likely to be involved in
characterizing the two distributions, a natural noninformative choice for the
direction probabilities is
$
\lambda_j(A,g) = 1/p$ for $j=1, \ldots, p$,
and  $g \in \{d,m\}$.
Finally, a standard  noninformative choice for the pseudo-counts is 
$\alpha^{j}_{t,l}(A,g) = Q_0(A^j_l)/Q_0(A)$ and $\alpha^{j}_{t,r}(A,g) =
Q_0(A^j_r)/Q_0(A)$, respecting the centering condition.

\section{Numerical examples}\label{seq:examples}

In this section we provide three numerical examples.
 The first two are simulated and the last a real flow cytometry data set.
We compare the performance of DMMT for two-sample testing---using
  $\Pr(H_0|\vect{x}_1,\vect{x}_2) $ 
as a test statistic---to those of several other state-of-the-art methods. 
We then illustrate how to identify and summarize differential structures using
the strategy given in Section~\ref{subsec:identify_difference}. The
dimensionalities of the three examples are one, two and seven. For all these
examples, we use a uniform baseline and follow the general recipe
for prior specification given in Section~\ref{sec:prior}.  In particular, we set
$\beta = 0.3$,
$\gamma=0.2$ and $\rho_{0,d}=1$, and we let $\rho_{g,s}(k)=1$ for  level $k=12$.
To identify the differential structure we consider the representative partition
tree
with $\delta^*=0.8$. 
We use the range of the data points in each dimension to define the
hyper-rectangle $\Omega$.

\subsection{ Example 1}\label{sec:example1}
In Example 1 we consider the following two-sample problems in $\mathbb{R}$. 
For each problem we generate 1,000 datasets, and for each dataset we construct a
corresponding ``null'' dataset by randomly permuting the labels of the two
groups.
\begin{enumerate}
\item Local shift difference ($n_1 = n_2 = 200$):
$
 X_{1}   \sim 0.9 \mathcal{N}(0.2, 0.05^2) + 0.1 
\mathcal{N}(0.9,
0.01^2)$, $
X_{2}    \sim 0.9  \mathcal{N}(0.2, 0.05^2) + 0.1 
\mathcal{N}(0.88,
0.01^2)$.
 \item Local dispersion difference ($n_1 = n_2 = 200$):
$
  X_{1}   \sim 0.9 \mathcal{N}(0.2, 0.05^2) + 0.1
\mathcal{N}(0.8,
0.01^2)$, 
$
X_{2}    \sim 0.9 \mathcal{N}(0.2, 0.05^2) + 0.1
\mathcal{N}(0.8,
0.04^2)
$.
\item Global shift difference ($n_1 = n_2 = 100$):
$
  X_{1}  \sim  \mathcal{N}(-0.5, 2^2), \quad
 X_{2}  \sim  \mathcal{N}(0.5, 2^2)$.
\item Global dispersion difference  ($n_1 = n_2 = 50$):
$
   X_{1}  \sim  \mathcal{N}(0, 1^2), \quad
 X_{2}  \sim  \mathcal{N}(0, 2^2).
$
\end{enumerate}
In the first row of  Figure \ref{fig:example1a} we plot the pair of density
functions for each scenario. In the
first two scenarios the difference is located in a small region of the sample,
while in the last two the difference is global.
In the second row of  Figure \ref{fig:example1a} we compare the ROC curves of
five different
statistics for testing the 
hypothesis that the two distributions are identical. The other four
test statistics are  
the $k$-nearest neighbors test (KNN) (\cite{schilling_1986}, and
\cite{henze_1988}), the Cram\'{e}r test \citep{baringhaus_franz_2004}, co-OPT
\citep{ma_wong_2011} and \cite{holmes_etal_2012}'s PT Bayes factor.
We use the  R package \texttt{MTSKNN} \cite[]{MTSKNN_package} for the KNN test
and the  R package \texttt{cramer} \cite[]{cramer_package} for the Cram\'{e}r
test.
 
The DMMT outperforms the other methods in the local difference scenarios and
behaves well in the global difference ones.
KNN behaves well in the local difference scenarios, but is not
efficient in the global scenarios. Cram\'{e}r, instead, is good
in testing larger-scale differences, but performs extremely poorly for local
differences.

\begin{figure}
\centering

\makebox{\includegraphics[width = 1.0\textwidth]{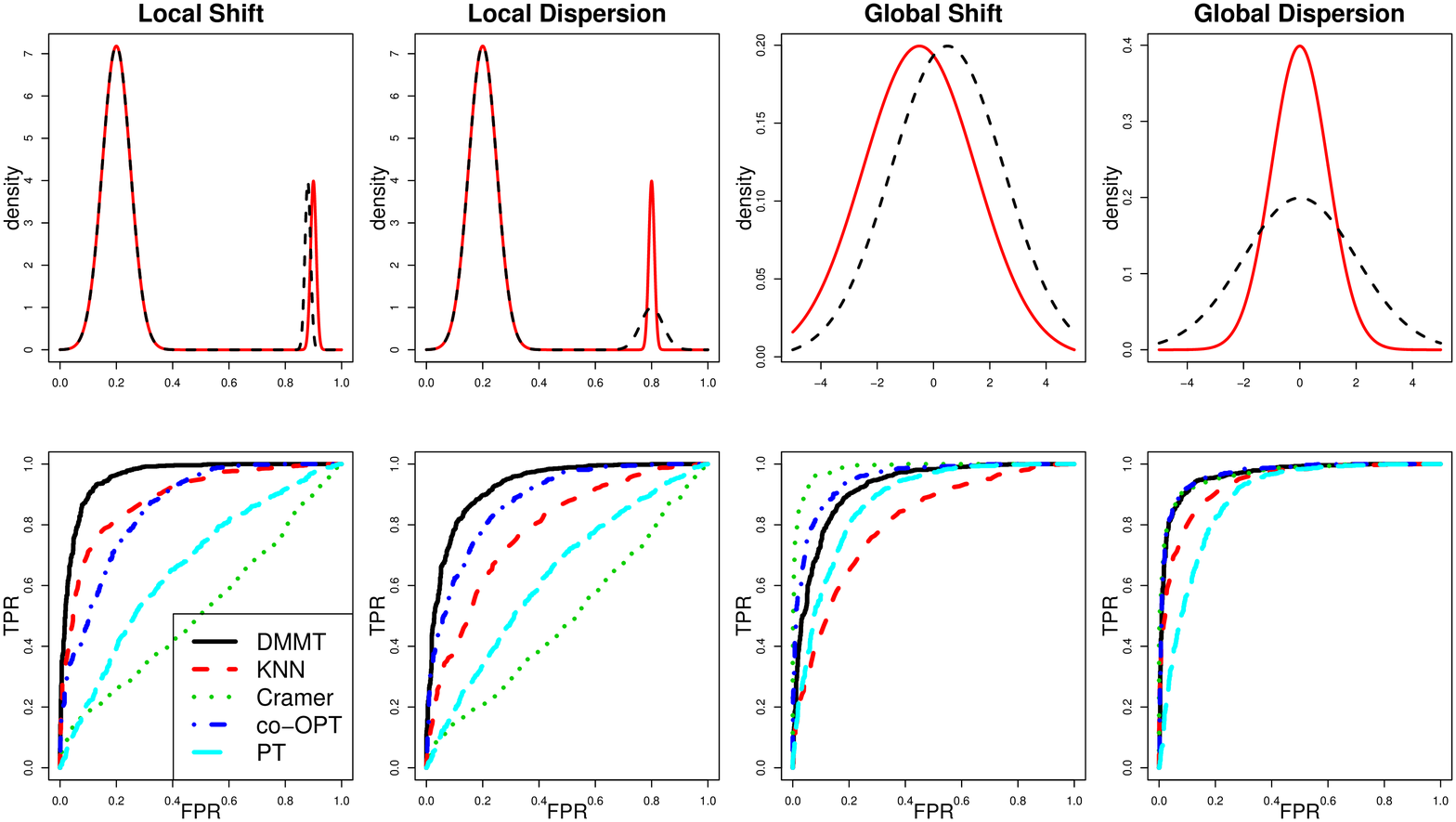}}
\caption{Two-sample problems in $\mathbb{R}$. First row: 
densities of the two distributions under the different scenarios - Sample 1 red
solid; Sample 2 black dashed. Second
row: the ROC curves for each of the testing method considered -
DMMT black solid; KNN red dash; Cram\'{e}r green dotted; co-OPT blue dotted
dash; PT pale blue long dash.}
\label{fig:example1a}

\vspace{.5cm}

\makebox{\includegraphics[width = 1\textwidth]{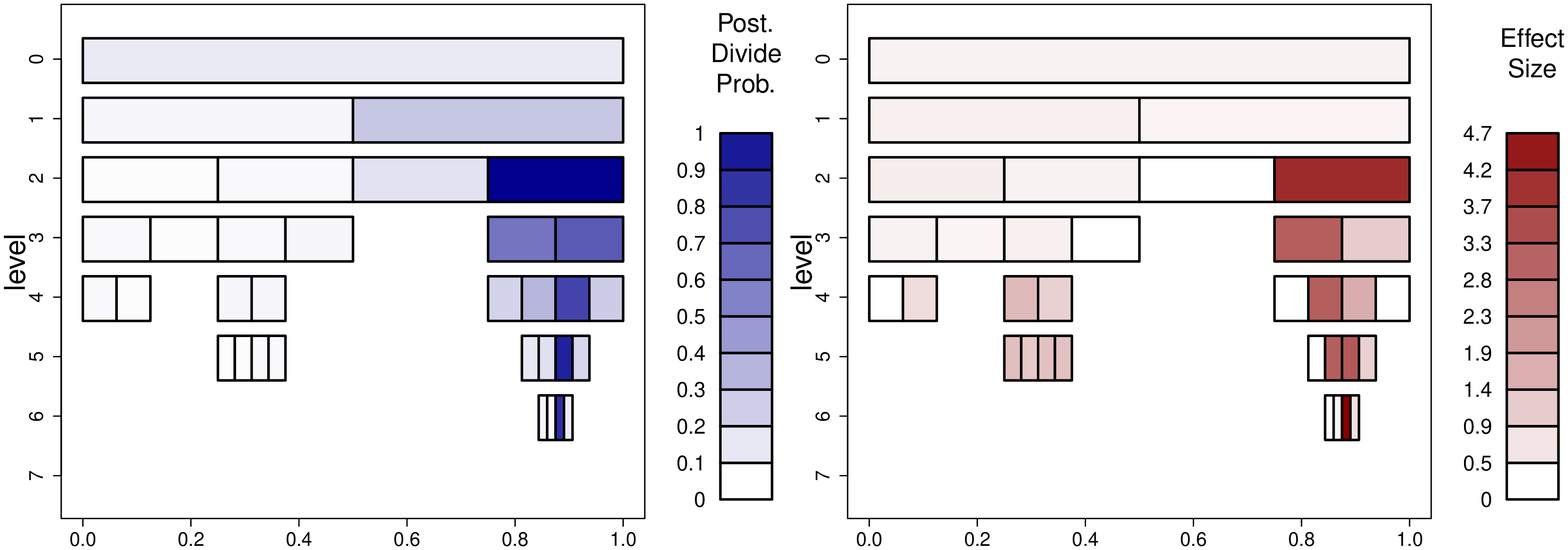}}
\caption{Nested sequence of partitions for the local shift difference scenario
in
$\mathbb{R}^1$. On the left, for each region
the dark/light blue represents 
 high/low posterior probability of being on the divide state. 
On the right, for each region
the dark/light red represents 
 high/low effect size. }
\label{fig:example1b}

\end{figure}
In Figure \ref{fig:example1b} we illustrate how to identify differences at 
multiple resolution levels
using the posterior representative partition tree for the local shift difference
scenario.
On the left plot we use a blue-scale to represent 
the marginal posterior probability of the divide state 
$\rho_d^*(\cdot|\vect{x}_1,\vect{x}_2)$ on each node; while on the right plot,
we use a
red-scale to visualize the effect size. 
The
two distributions share the same normal component on
the left, while the difference is located on the right, which is effectively
captured in the posterior summary. The spatial clustering of the difference is
also reflected in the patch of dark blue and red in the lower right corner of
the two plots.

In the Supplementary Material we provide an analysis of sensitivity to the
prior choice, showing that the
results are robust to the prior specification. Similar results are obtained for
the other scenarios and examples and so are not reported.

\subsection{Example 2}
We consider the following two-sample problems in $\mathbb{R}^2$. 
 For each problem we again generate 1,000 datasets, and
for each dataset we construct a corresponding ``null'' dataset by randomly
permuting the labels of the two
groups.
\begin{enumerate}
\item Local shift difference ($n_1 = n_2 = 400$):
$   X_{1}  \sim  p_1  \mathcal{N}_2(\mu_1, \Sigma_1) +  \sum_{k= 2}^{5} p_k
 \mathcal{N}_2(\mu_k, \Sigma_k)$, 
$   X_{2}  \sim  p_1  \mathcal{N}_2(\mu_1 + \delta, \Sigma_1) + 
\sum_{k= 2}^{5} p_k \mathcal{N}_2(\mu_k, \Sigma_k)$,
where $\delta = (1,1)$, while $p_k$, $\mu_k$ and $\Sigma_k$ are
provided in Appendix \ref{ap:bivariate_example}.
\item Local dispersion difference ($n_1 = n_2 = 400$):
$   X_{1}  \sim  p_1  \mathcal{N}_2(\mu_1, \Sigma_1) +  \sum_{k= 2}^{5}
p_k \mathcal{N}_2(\mu_k, \Sigma_k)$, 
 $  X_{2}  \sim  p_1  \mathcal{N}_2(\mu_1, 5\Sigma_1) + 
\sum_{k= 2}^{5} p_k
 \mathcal{N}_2(\mu_k, \Sigma_k)$,
where $p_k$, $\mu_k$ and $\Sigma_k$ are
provided in Appendix \ref{ap:bivariate_example}.
 \item Global shift difference ($n_1=n_2 = 100$):
$
   X_{1}  \sim    \mathcal{N}_2(0, \Sigma), \quad
   X_{2}  \sim    \mathcal{N}_2(\delta, \Sigma),
$
 where  $\delta = (1,0)$, while $\Sigma$ is provided in Appendix
\ref{ap:bivariate_example}.
\item Global dispersion difference ($n_1=n_2 = 50$):
$
   X_{1}  \sim   \mathcal{N}_2(0, I_2), \quad
    X_{2}  \sim   \mathcal{N}_2(0, 3  I_2 )
$.
\end{enumerate}
In Figure \ref{fig:example2a}  we plot the ROC curve for DMMT
and the other
methods. The results are similar to what we have obtained in the 1D example. 
The performance of DMMT dominates two of the competing
methods---KNN and PT---in the sense that DMMT is at least as good in
all four scenarios. The performance gain over the co-OPT is largely due to
incorporation of spatial clustering through Markov dependence, while the gain
over the PT is due to both the Markov dependence and the data-adaptive
partition sequence. On the other hand, DMMT does substantially better than
Cram\'{e}r for local differences, while Cram\'{e}r is more powerful, though to a
lesser extent, for the global differences. 

\begin{figure}
 \centering

\makebox{\includegraphics[width = 1.0\textwidth]{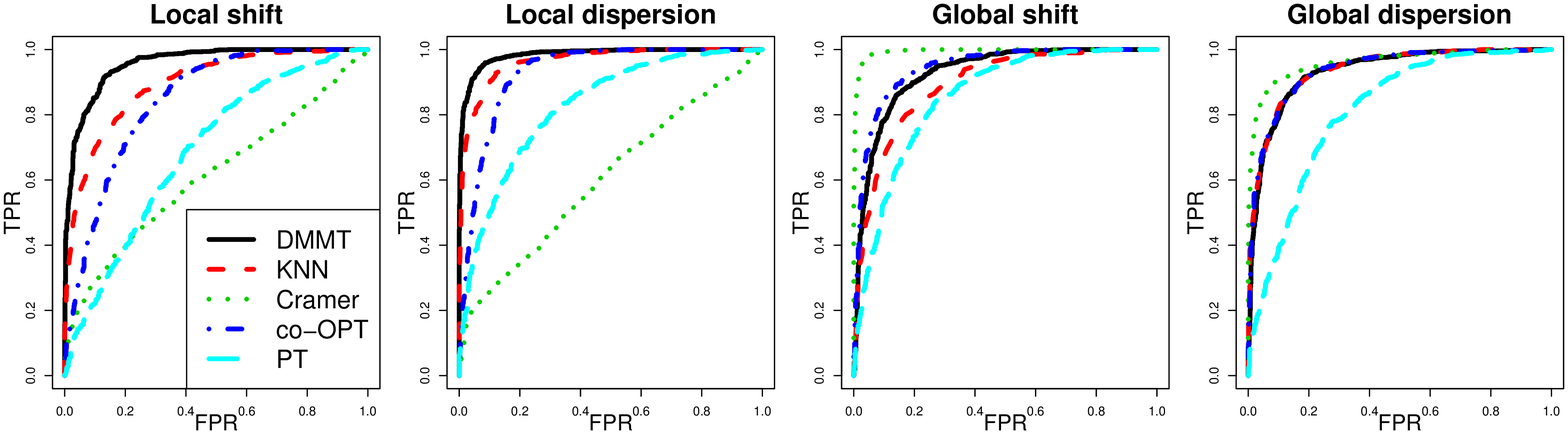}}
\caption{Two-sample problems in $\mathbb{R}^2$. ROC curves for each of the
testing method considered -
DMMT black solid; KNN red dash; Cram\'{e}r green dotted; co-OPT blue dotted
dash; PT pale blue long dash.}
\label{fig:example2a}

\vspace{.5cm}

\makebox{\includegraphics[width = 1.0\textwidth]{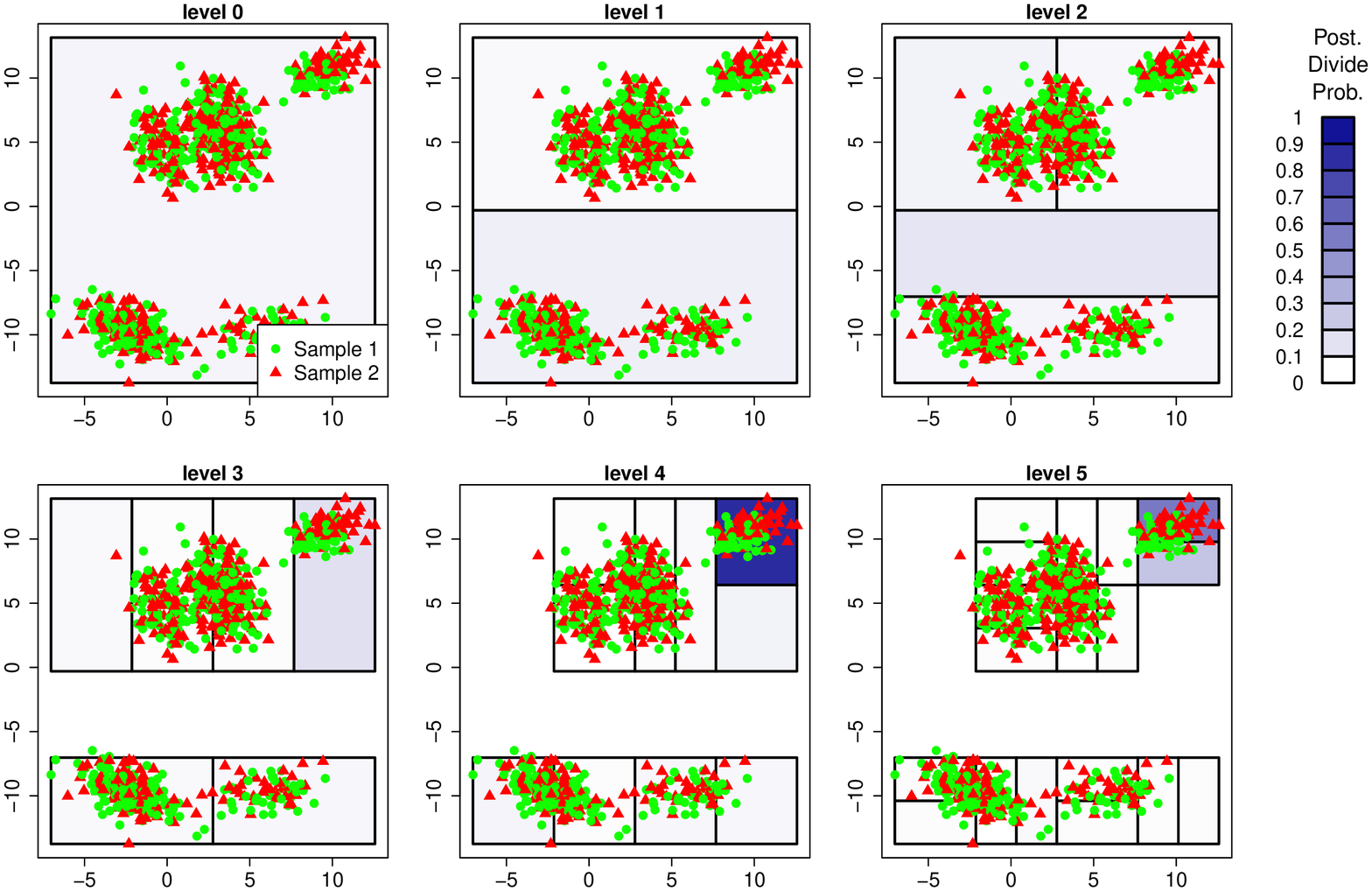}}
\caption{The representative partition tree for a draw from the local shift
difference
scenario in $\mathbb{R}^2$. Sample 1 green dots; Sample 2 red triangles. For
each region
the dark/light blue represents 
 high/low posterior probability of being on the divide state. 
}
\label{fig:example2b}

\end{figure}

For a typical data set in the local shift difference
scenario, we again illustrate how to identify the differential structure. 
We plot in Figure \ref{fig:example2b} 
the representative partition tree. 
For each region, we show 
the marginal posterior probability of the divide state 
$\rho_d^*(\cdot|\vect{x}_1,\vect{x}_2)$.
These regions correctly
identify the local
differences between the two distributions.

\subsection{A 7-dimensional flow cytometry dataset}

Flow cytometry is a popular laser technology for 
measuring  the protein levels of single cells on thousands of cells.
In this example we have two blood samples from the same patient,
where each sample contains over $300,000$ cells, and for each cell the
following 7 markers are measured: 
FSC-A, FSC-H, SSC-A, Dext, CD4, CD8 and Aqua. 
 In one of the two samples some cells have been
transfected via electroporation with a T cell receptor gene specific for
Tyrosinase  (see \cite{singh_etal_2013} for further details).
The dexter conjugated with a fluorescent dye detects
Tyrosinase-spefic CD8 T cell receptors, and  a
higher concentration of transfected cells is expected in the population of
cells that are both CD8 and Dext high.

The primary goal of this analysis is to identify cell subpopulations
 that differentiate the two cell samples.
For this reason the existing two sample tests such as KNN and Cram\'{e}r are not
useful here as one cannot pinpoint the difference with the tests.

We apply our DMMT process to model the underlying distributions. 
The posterior probability for the two samples to be equal is virtually 0. Using
the representative partition tree we identify 
several differential regions with significant two-sample differences.
In Figure \ref{fig:example3a} we plot five 2-dimensional projections of the 
data, highlighting in yellow the
differential region with 
the largest effect size among the regions with $\rho_d^*(\cdot|\vect{x}_1,
\vect{x}_2) > \delta^*$. In particular, $\rho_d^*(\cdot|\vect{x}_1,
\vect{x}_2) = 0.998$ and the effect size is $4.35$. 
We ``smeared'' out the cells that do not fall into our detected region.
The volume of the region is $1/64$ of the entire sample space and  contains
respectively $0.006\%$  and $0.220\%$ of the two cell populations (indicated
with
red dots and blue triangles). 
This reported difference is scientifically validated because a difference
 is expected for CD8 and Dext high cells. We  observe a probable
cluster  of transfected cells (blue triangles)
in the rightmost plot where the data 
are projected along the CD8 and Dext markers.

\begin{figure}
 \centering
\makebox{\includegraphics[width =
1.0\textwidth]{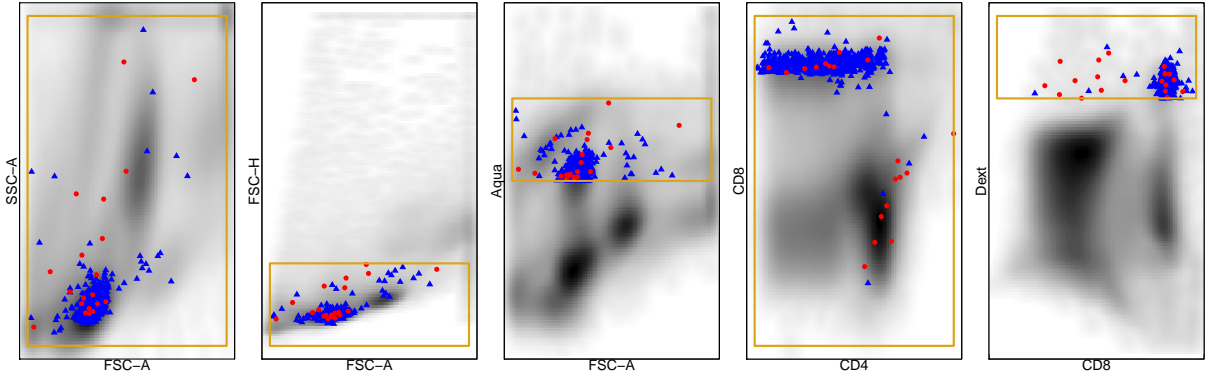}}
\caption{Five projections of  the flow cytometry dataset. 
For each projection the yellow rectangle highlights the differential region
with the largest effect size among the regions with $\rho_d^*(\cdot|\vect{x}_1,
\vect{x}_2) > \delta^*$. The red dots and the blue triangles represent
respectively the  normal and transfected cells within the differential region.
In the plot on the far right we observe a cluster of transfected cells (blue
triangles).}
\label{fig:example3a}
\end{figure}

\section{ Conclusion}\label{seq:conclusion}
In this work we have introduced a new multi-resolution Bayesian framework for
testing and identifying two-sample differences 
based on a nonparametric process called the DMMT. 
This process uses Markov dependence to incorporate spatial clustering of
differences, and randomized partitioning to achieve a data-adaptive
multi-resolution partition sequence. We have provided a recipe for inference
based on forward-summation-backward-sampling type of recursive computation, and
have showed how two-sample testing and identification of the differential
structure
can be carried out also based on recursive calculation. We have also established
the large support property of the DMMT and proved the asymptotic consistency of
the two-sample test based on this process.

The design of the DMMT process does not require either dimensionwise or dyadic
 partition sequences. Neither does it require the partition point to be in the
middle of support in each dimension. More general partition schemes have been
earlier considered by \cite{ma_wong_2011}, and the DMMT can be directly extended
to incorporate those more general partition systems. In this work we have
deliberately chosen to introduce the DMMT with dimensionwise dyadic partitions
to simplify the presentation and thereby highlighting the more unique features
of the DMMT in comparison to other existing multi-resolution tests. 

For the same reason we have introduced the DMMT in the context of comparing 
two data samples. The design of DMMT  allows comparison across $k$
samples
($k>2$)---simply let the generative process of DMMT generate $k$
distributions $(Q_1, \ldots,Q_k)$ simultaneously by drawing $k$ assignment
variables $Y_1(A),\ldots,Y_k(A)$ in the probability assignment step. The
inferential recipe, computational algorithm, and theoretical results all remain
valid.

Finally, we compare the computational efficiency of some of the
methods considered in this work. The computation was executed on a
single Intel(R)
Core(TM) i7-3820 core with a 3.60GHz CPU clock.  In the one-dimensional local
shift
example
the median CPU times are $0.001s$ (KNN),   $0.004s$ (DMMT) and  $0.482s$
(Cram\'{e}r).
In the two-dimensional local shift  example the median CPU times are
 $0.004s$ (KNN),   $0.073s$ (DMMT) and  $1.983s$ (Cram\'{e}r). 
In the flow cytometry example the CPU times are $3085s$ (KNN), 
 $7307s$ (DMMT), while we were unable to fit  Cram\'{e}r.  (Note that although
one can fit KNN to the data, it is not useful here because it cannot identify
where the difference is located.).
The computational complexity of DMMT is  linear in $n$, while those for  both
KNN
and Cram\'{e}r are quadratic. However, given the exponential nature of DMMT in
$p$ and $k$, efficient approximations are needed to scale it to high-dimensional
problems.

\section*{Software}
An R package implementing our method is available at
\href{https://github.com/jacsor/DMMT}{https://github.com/jacsor/DMMT}. 

\section*{Acknowledgment}
LM's research is supported by NSF grant DMS-1309057.
Part of this work was completed while LM was a fellow, and JS a graduate fellow
at Statistical and Applied Mathematical Sciences Institute (SAMSI).
The authors want to thank Cliburn Chan for providing the flow cytometry data,
which was made possible by grants from the
Wallace Coulter Foundation and the NIH (P50-5P30 AI064518). 

\section*{ Appendix}

\appendix
\makeatletter   

\section{ Numerical example 2}\label{ap:bivariate_example}
\begin{enumerate}
 \item Local shift difference: 
$(p_1, \ldots, p_5) = (0.11, 0.16, 0.25, 0.39,
0.09), \delta=(1.0, 1.0),\\ \mu_1 = (9.0, 9.9),  \mu_2 = (0.0, 4.4), \mu_3 =
(-2.3, -9.7), \mu_4 = (3.4, 5.9), \mu_5 = (5.8, -9.5), \\
( \Sigma_1(1,1), \Sigma_1(1,2), \Sigma_1(2,2) ) = (2.9, 0.5, 1.1),
( \Sigma_2(1,1), \Sigma_2(1,2), \Sigma_2(2,2) ) = (1.2, -0.6, 2.8),\\
( \Sigma_3(1,1), \Sigma_3(1,2), \Sigma_3(2,2) ) = (2.3, -1.0, 1.7),
( \Sigma_4(1,1), \Sigma_4(1,2), \Sigma_4(2,2) ) = (1.1, -0.4, 2.9),\\
( \Sigma_5(1,1), \Sigma_5(1,2), \Sigma_5(2,2) ) = (3.0, 0.2, 1.0)$. 
 \item Local dispersion difference: 
$(p_1, \ldots, p_5) = (0.19, 0.08, 0.33, 0.27,
0.13),  \mu_1 = (0.9, -7.2), \mu_2 = (-5.7, 3.3),  \mu_3 =
(-6.3, -2.1), \mu_4 = (7.5, -3.1), \mu_5 = (-3.1, 9.5), \\
( \Sigma_1(1,1), \Sigma_1(1,2), \Sigma_1(2,2) ) = (0.5, -0.1, 0.3),
( \Sigma_2(1,1), \Sigma_2(1,2), \Sigma_2(2,2) ) = (1.3, 0.7, 2.7),\\
( \Sigma_3(1,1), \Sigma_3(1,2), \Sigma_3(2,2) ) = (1.0, -0.3, 3.0),
( \Sigma_4(1,1), \Sigma_4(1,2), \Sigma_4(2,2) ) = (2.9, 0.5, 1.1),\\
( \Sigma_5(1,1), \Sigma_5(1,2), \Sigma_5(2,2) ) = (2.4, -0.9, 1.6)
$. 
\item Global shift difference: 
$ ( \Sigma(1,1), \Sigma(1,2), \Sigma(2,2) ) = (2.9, 0.4, 1.1)$.
\end{enumerate}

\bibliographystyle{Chicago}
\bibliography{citations_database}{}

\clearpage

\appendix
\numberwithin{equation}{section}
\setcounter{section}{2}

\setcounter{mylemma}{0}
\setcounter{mytheorem}{0}

\section*{Supplementary Material}

\subsection{Proofs}

\begin{mytheorem}\label{thm:large_support}

\end{mytheorem}
\begin{proof}
First, assume without loss of generality that $\tilde{f}_t$ for $t=1,2$ are
uniformly continuous because other densities can be approximated arbirtrarily
well by uniformly continuous ones. Let
$$
\delta_t(\epsilon) = \text{sup}_{|x=y|<\epsilon}|\tilde{f}_t(x)-\tilde{f}_t(y)|.
$$
Then, by uniformly continuity, $\delta_t(\epsilon) \downarrow 0 $ as $\epsilon
\downarrow 0$ for $t=1,2$. By Condition (1), for any $\epsilon > 0 $ there is a
partition of $\Omega = \cup_{i=1}^I A_i$ such that the diameter of each $A_i$
is less than $\epsilon$ and $A_i \in \mathcal{A}^{(\infty)}$. By Conditions (2)
and (3) there is a positive probability that this partition will arise during
the partitioning process in a finite number of steps. 
Additionally, there is positive probability that the process will ``merge'' on
each $A_i$. Then, the rest of the proof follows from the proof of Theorem~3 in
\cite{ma_wong_2011}.
\end{proof}

\begin{mylemma}\label{lemma:1}
\end{mylemma}
\begin{proof}
For a set $A$ that arises during the partitioning process, consider
an observation  $x \in A$ of $\vect{x}_1$. 
Let $q_1(x|A)= q_1(x)/Q_1(A)$ and $q_0(x|A)= q_0(x)/Q_0(A)$, where $q_1 =
dQ_1/d\mu$ and $q_0 =
dQ_0/d\mu$. That is, $q_1(x|A)$ and $q_0(x|A)$  are the conditional density on
$A$ for $Q_1$ and $Q_0$, respectively. Then, 
$$
q_{1}(x | A)  =
\left\{
\begin{array}{ll}
 q_{0}( x | A )  &  \text{ if } S(A) = s \\
  \big{(} Y_{1}(A) \big{)}^{1(x
\in A^{J(A)}_l )}  \big{(} 1-Y_{1}(A) \big{)}^{1(x
\in A^{J(A)}_r )}  & \\ 
 \quad \cdot \prod\limits_{i\in\{l,r\}} q_{1}(x | A^{J(A)}_i)   
& \text{ if
} S(A) \in
\{d,m\}.
\end{array}
\right.
$$
Define $q_1(\vect{x}_1|A) = \prod_{x \in \vect{x}_1} q_{1}(x | A)$ and
$q_0(\vect{x}_1|A) = \prod_{x \in \vect{x}_1} q_{0}(x | A)$, then
$$
q_{1}( \vect{x}_1 | A)  =
\left\{
\begin{array}{ll}
 q_0(\vect{x}_1|A)  & \text{if } S(A) = s, \\
  \big{(} Y_{1}(A) \big{)}^{n_1(A^{J(A)}_l) }  \big{(} 1-Y_{1}(A)
\big{)}^{n_1(A^{J(A)}_r) } & \\
\quad \cdot \prod\limits_{i\in\{l,r\}} q_{1}(\vect{x}_1 |
A^{J(A)}_i) & \text{if
} S(A) \in
\{d,m\},
\end{array}
\right.
$$
where $n_1(A) = | \{x_{1,i} : x_{1,i} \in A,\; i=1,2, \ldots, n_1\} |$.
Similarly, define $q_2(\vect{x}_2)$, $q_2(\vect{x}_2|A)$ and
$n_2(A)$.
Finally,
define $q(\vect{x}_1,\vect{x}_2|A) = q_2(\vect{x}_1|A) \cdot q_2(\vect{x}_2|A)$,
then
$$
q( \vect{x}_1,\vect{x}_2 | A)  =
\left\{
\begin{array}{ll}
 q_0(\vect{x}_1|A)\cdot q_0(\vect{x}_2|A)  & \text{if } S(A) = s, \\
 \prod\limits_{t=1,2}   \big{(} Y_{t}(A) \big{)}^{n_t(A^{J(A)}_l) }  \big{(}
1-Y_{t}(A)
\big{)}^{n_t(A^{J(A)}_r) } & \\
\quad \cdot \prod\limits_{i\in\{l,r\}}
q(\vect{x}_1,\vect{x}_2 |
A^{J(A)}_i) & \text{if
} S(A) \in
\{d,m\}.
\end{array}
\right.
$$
Note that $q( \vect{x}_1,\vect{x}_2 | A)$ is the conditional likelihood
associated to the
data points in $A$. 
In particular, when $S(A)=m$, then $Y_1(A)=Y_2(A)=Y(A)$,
and so we can write it
as
\begin{equation}\label{eq:cond_like}
q(\vect{x}_1,\vect{x}_2 | A)  =
\left\{
\begin{array}{ll}
 \prod_{x\in \vect{x}_1\cup \vect{x}_2} q_{0}( x | A )  & \text{if } S(A) = s,
\\
  \big{(} Y(A) \big{)}^{
n(A^{J(A)}_l )} \big{(} 1-Y(A) \big{)}^{
n(A^{J(A)}_r )} & \\
\quad \cdot \prod\limits_{i\in\{l,r\}} q ( \vect{x}_1,\vect{x}_2 |
A^{J(A)}_i)   & \text{if } S(A) = m, \\
 \prod_{t=1,2} \big{(} Y_t(A) \big{)}^{
n_t(A^{J(A)}_l )}\big{(} 1-Y_t(A) \big{)}^{
n_t(A^{J(A)}_r )}  & \\
\quad \cdot q( \vect{x}_1, \vect{x}_2 |
A^{J(A)}_i)   & \text{if } S(A) = d,
\end{array}
\right.
\end{equation}
where $n(\cdot) = n_1(\cdot) + n_2(\cdot)$. 

Define $Z(A,g,\vect{x}_1,\vect{x}_2)$ as the marginal likelihood conditional on
$A$, given
$E_1(A)$ and $S(A)=g$ for $g\in \mathcal{G}$. Then, we obtain  
\begin{equation}\label{eq:Phi}
\Phi(A,g,\vect{x}_1,\vect{x}_2) = \sum_{h \in \mathcal{G}} \rho_{g,h}(A)
Z(A,h,\vect{x}_1,\vect{x}_2).
\end{equation}
If $S(A)=s$, then $q_t(\cdot|A) = q_0(\cdot|A)$, and so 
$Z(A,s,\vect{x}_1,\vect{x}_2)=
\prod_{x \in \vect{x}_1\cup\vect{x}_2} q_0(x|A)$. Otherwise, if $S(A) = g$, for
$ g \in \{ d,m
\}$, by integration of the right hand-side of \eqref{eq:cond_like}, we obtain 
\begin{align*}
  Z(A,m,\vect{x}_1,\vect{x}_2) & = \sum \limits_{j=1}^{p} \lambda_j(A,m) E
\bigg{[}  \big{(} Y(A) \big{)}^{
n(A^{j}_l )} \big{(} 1-Y(A) \big{)}^{
n(A^{j}_r )} \bigg{]} \\ 
& \cdot \prod \limits_{i \in \{l,r\}} \int q(
\vect{x}_1,\vect{x}_2 |
A^{j}_i) \pi\big{(}dq|E_1(A^{j}_i), E_{2,m}(A^{j}_i)\big{)} \\
& = \sum \limits_{j=1}^{p} \lambda_j(A,m)  \dfrac{D(
\vect{\alpha}^j(A,m) +
\vect{n}_1^j(A) + \vect{n}_2^j(A)
)}{D( \vect{\alpha}^j(A,m) )} \cdot \prod \limits_{i \in \{l,r\}}  
\Phi(A^j_i,m,
\vect{x}_1,\vect{x}_2) \\
& = \sum_{j=1}^{p} Z_j(A,m,\vect{x}_1,\vect{x}_2),
\end{align*}
and
\begin{align*}
 Z(A,d,\vect{x}_1,\vect{x}_2) & = 
\sum \limits_{j=1}^{p} \lambda_j(A,d) E \bigg{[} \prod\limits_{t=1,2}\big{(}
Y_{t}(A) \big{)}^{
n_t(A^{j}_{l} )}
\big{(}1-
Y_{t}(A) \big{)}^{
n_t(A^{j}_{r} )} \bigg{]} \\
& \cdot \prod \limits_{i \in \{l,r\}} \int q(
\vect{x}_1,\vect{x}_2
|A^{j}_i) \pi\big{(}dq| E_1(A^{j}_i), E_{2,d}(A^{j}_i)\big{)} \\
& = 
\sum \limits_{j=1}^{p} \lambda_j(A,d) \prod \limits_{t=1,2}\dfrac{D(
\vect{\alpha}^j_t(A,d) + \vect{n}^j_t(A) )}{D( \vect{\alpha}^j_t(A,d)  )} \cdot
\prod \limits_{i \in \{l,r\}} 
\Phi(A^j_i,d, \vect{x}_1,\vect{x}_2)\\
& = \sum_{j=1}^{p} Z_j(A,d,\vect{x}_1,\vect{x}_2).
\end{align*}
\end{proof}

\begin{mylemma}

\end{mylemma}
\begin{proof}
The first claim holds because $\Phi(A,g,\vect{x}_1,\vect{x}_2)$ is the marginal
likelihood, and so is equal to $1$ if there are no data-points in $A$.
The second claim follows from the DMMT self-similarity. For any region $A$ that
arises during the partitioning process, the DMMT restricted to $A$ is a
DMMT process with sample space $\Omega=A$,  parameters 
restricted to $A$ and its descendants, and baseline measure $Q_0(\cdot|A)$.
The centering condition assumption implies that $q_0(\cdot|A)$ is the predictive
density. For a single data-point $x$ in $A$, then $\Phi(A,g,\vect{x})$ is the
predictive density at $x$. 
\end{proof}

\begin{mytheorem}\label{thm:2}
 
\end{mytheorem}
\begin{proof}
The proof is based on the results of Lemma \ref{lemma:1}.
 Note that the right-hand side of \eqref{eq:Phi} 
is the sum of over $h \in \mathcal{G}$ of
$$
\Pr \{ \text{generate } \vect{x}_1, \vect{x}_2 | S(A)=h, S(\text{parent}(A))=g 
\} 
\Pr
\{ S(A)=h| S(\text{parent}(A))=g \},
$$ 
for $g \in \mathcal{G}$ and $A \in \mathcal{A}^{(\infty)}$. Thus,
$$
\Pr \big( S(A)=g | S(\text{parent}(A))=h, \vect{x}_1, \vect{x}_2 \big) = 
\rho_{g,h}(A|\vect{x}_1,\vect{x}_2) = \rho_{g,h}(A) \dfrac{
Z(A,h,\vect{x}_1,\vect{x}_2)}{\Phi(A,g,\vect{x}_1,\vect{x}_2)}, 
$$
for all $A \in \mathcal{A}^{(\infty)}$, $g,h\in\mathcal{G}$.
Similarly, $Z_j(A,g,\vect{x}_1, \vect{x}_2)$ represents 
$$
\Pr \{ \text{generate } \vect{x}_1, \vect{x}_2 | S(A)=g, \text{ split direction
}
j \}
\Pr 
\{  \text{ split direction } j | S(A)=g
  \},
$$ 
for $g \in \mathcal{G}$, $j=1, \ldots, p$ and $A \in \mathcal{A}^{(\infty)}$.
Thus,
$$
\Pr \big( J(A) = j | S(A)=g, \vect{x}_1,\vect{x}_2 \big) = \dfrac{
Z_j(A,g,\vect{x}_1,\vect{x}_2) }{
Z(A,g,\vect{x}_1,\vect{x}_2) },
$$
for all $A \in \mathcal{A}^{(\infty)}$, $g\in \{ d,m\}$ and
$j=1,\ldots,p$. Finally, since the Beta distribution is conjugate to the
Multinomial distribution, we have
$$
\alpha_t(A,g|\vect{x}_1,\vect{x}_2) = \left\{ \begin{array}{ll}
            \alpha_t(A,d)  + n_t(A) & \text{if } g= d \\
            \alpha_t(A,m)  + n_1(A) + n_2(A) & \text{if } g= m, 
                           \end{array}
\right.
$$
for all $A \in \mathcal{A}^{(\infty)}$, $g\in \{d,m\}$ and $t=1,2$. 
\end{proof}

\begin{mylemma}

\end{mylemma}
\begin{proof}
By self-similarity and conjugacy, we know that
$$
(Q_1, Q_2)| (\vect{x}_1,\vect{x}_2,E_1(A),E_{2,g}(A))
$$
 is still a DMMT with
initial state $\rho_{0,g}=1$ and parameters defined by Theorem \ref{thm:2}.
Consider the event
$$
\{ q_1(\cdot|A)  = q_2(\cdot|A) \} = \bigcup_{ h \in \mathcal{G}} \big\{
q_1(\cdot|A) =  q_2(\cdot|A),
 S(A)=h \big{\}},
$$
For $S(A)=s$ we have
$$
\{ q_1(\cdot|A)  =  q_2(\cdot|A),
 S(A)=s \big\}  = 
\{ q_1(\cdot|A) =  q_2(\cdot|A) = q_0(\cdot|A),
 S(A)=s \big\},
$$
and so
$$
\int 1\{ q_1(\cdot|A)  =  q_2(\cdot|A),
 S(A)=s \big\}  \pi(dq|\vect{x}_1,\vect{x}_2, E_1(A),E_{2,g}(A)) =
\rho_{g,s}(A|\vect{x}_1,\vect{x}_2).
$$
For $S(A)=m$ we have
$$
\{ q_1(\cdot|A)  =  q_2(\cdot|A),
 S(A)=m \big\}  = \bigcup_{ j=1}^{p}  \{ q_1(\cdot|A) =  q_2(\cdot|A),
 S(A)=m, J(A)=j \big\}, $$
 where
\begin{align*}
 \{  q_1(\cdot|A) & =  q_2(\cdot|A),
 S(A)=m, J(A)=j \big\}  \\
& =  \{ Y(A) q_1(\cdot|A_l^j) + (1-Y(A))
q_1(\cdot|A_r^j) \\ & =  Y(A) q_2(\cdot|A_l^j) + (1-Y(A))
q_2(\cdot|A_r^j), 
  S(A)=m, J(A)=j \big\} \\
 & =\{  q_1(\cdot|A_l^j) =  q_2(\cdot|A_l^j), q_1(\cdot|A_r^j) = 
q_2(\cdot|A_r^j),
 S(A)=m, J(A)=j \big\}
 \end{align*}
and so
 \begin{align*}
   \int 1 \{ q_1(\cdot|A) & =  q_2(\cdot|A),
 S(A)=m, J(A)=j \big\} \pi(dq|\vect{x}_1,\vect{x}_2, E_1(A),E_{2,g}(A)) \\
 & = \rho_{g,m}(A|\vect{x}_1,\vect{x}_2)\lambda_j(A,m|\vect{x}_1,\vect{x}_2) \\
& \quad \cdot \prod_i \int 1 \{ q_1(\cdot|A_i^j) = q_2(\cdot|A_i^j) \}
\pi(dq|\vect{x}_1,\vect{x}_2, E_1(A_i^j),E_{2,m}(A_i^j)) \\
 & = \rho_{g,m}(A|\vect{x}_1,\vect{x}_2)\lambda_j(A,m|\vect{x}_1,\vect{x}_2) 
\prod_i \Psi(A_i^j,m,\vect{x}_1,\vect{x}_2).
 \end{align*}
Finally, for $S(A)=d$ we have
$$
\{ q_1(\cdot|A)  =  q_2(\cdot|A),
 S(A)=d \big\}  = \bigcup_{ j=1}^{p}  \{ q_1(\cdot|A) =  q_2(\cdot|A),
 S(A)=d, J(A)=j \big\}, $$
 where
\begin{align*}
\{  q_1(\cdot|A) & =  q_2(\cdot|A),
 S(A)=d, J(A)=j \big\}  \\
& =  \{ Y_1(A) q_1(\cdot|A_l^j) + (1-Y_1(A))
q_1(\cdot|A_r^j)  \\
& =  Y_2(A) q_2(\cdot|A_l^j) + (1-Y_2(A))
q_2(\cdot|A_r^j), 
 S(A)=m, J(A)=j \big\} \\
 & =\{  Y_1(A) = Y_2(A), q_1(\cdot|A_l^j) =  q_2(\cdot|A_l^j), \\
& \quad
q_1(\cdot|A_r^j)
= 
q_2(\cdot|A_r^j),
 S(A)=m, J(A)=j \big\},
 \end{align*}
 but the event has null probability since $Y_1(A),Y_2(A)$ are independent and
absolutely continuous, then
$$
  \int 1 \{ q_1(\cdot|A) =  q_2(\cdot|A),
 S(A)=d, J(A)=j \big\} \pi(dq|\vect{x}_1,\vect{x}_2, E_1(A),E_{2,g}(A)) =0.
 $$

\end{proof}

\vspace{2cm}

We observe two independent groups of i.i.d.\ samples $\vect{x_1} = (x_{1,1},
\ldots, x_{1,n_1})$ and
$\vect{x_2} = (x_{2,1}, \ldots, x_{2,n_2})$ from two distributions $Q_1$ and
$Q_2$, where $n = n_1+n_2
\rightarrow \infty$, and  $n_1/n\rightarrow
\beta$ for some $\beta \in (0,1)$. 
Let $(Q_1,Q_2)$ have a
DMMT$(\vect{\rho},\vect{\lambda},\vect{\alpha},
Q_0)$.
The two unknown distributions are $P_1$ and $P_2$. Define 
$P = \beta P_1 + (1-\beta) P_2$, $p_t^j = P_t^j(A_l^j|A)$,
$p^j = P(A_l^j|A) $, $q^j=
Q_0(A^j_i) / Q_0(A)$,
$\hat{p}_t^j = n_t(A_l^j)/n_t(A)$  and $\hat{p}^j = n(A_l^j)/n(A)$ for
$t=1,2$ and $j=1, \ldots, p$.

\begin{mylemma}\label{lemma:lambda_md}
For $A \in \mathcal{A}^{(\infty)}$ and $j=1, \ldots, p$, define
$$
 \Lambda_{m,d}(A,j, \vect{x}_1, \vect{x}_2) = \dfrac{ D(
\vect{\alpha}^j(A,m) + \vect{n}^j_1(A) + \vect{n}^j_2(A) )/D(
\vect{\alpha}^j(A,m)  ) }{ \prod_{t=1,2} D(
\vect{\alpha}^{j}_t(A,d) + \vect{n}^{j}_t(A) )/D( \vect{\alpha}^{j}_t(A,d)  )}.
$$
Then, 
\begin{enumerate}
 \item $\log \Lambda_{m,d}(A,j,\vect{x}_1, \vect{x}_2) / \log n \overset{p}{\to}
\eta' $ for some $\eta'>0$ if $p_1^j=p_2^j$.
  \item $\log \Lambda_{m,d}(A,j,\vect{x}_1, \vect{x}_2) / n \overset{p}{\to}
\eta'' $ for some $\eta'' < 0$ if
$p_1^j\neq p_2^j$.
\end{enumerate}
\end{mylemma}
\begin{proof}
The case $p_1^j = p_2^j$ follows by Theorem 1 in \cite{holmes_etal_2012}, but we
report the proof since it used to show the case $p_1^j \neq p_2^j$. 
 Using Stirling's formula
\begin{equation}\label{eq:stirling}
D(x,y)=\dfrac{\Gamma(x)\Gamma(y)}{\Gamma(x+y)} \simeq\sqrt{2\pi}\dfrac{
x^{x-1/2}
y^{y-1/2}}{ (x+y)^{x+y-1/2} }, \text{
for large }x,y,
\end{equation}
we can approximate the ratio
\begin{align*}
\Lambda_{m,d}(A,j,\vect{x}_1, \vect{x}_2) & \simeq  \dfrac{\prod_{t=1,2} D(
\vect{\alpha}^{j}_t(A,d)  )  }{ D(\vect{\alpha}^j(A,m)  )
  } \cdot \dfrac{1}{\sqrt{2\pi}}   \cdot \dfrac{
(\hat{p}^j)^{\alpha(A^{j}_l,m)
- 1/2)} (1 - \hat{p}^j)^{ \alpha(A^{j}_r,m) - 1/2)}  }{ \prod_t
(\hat{p}_t^{j})^{\alpha_t(A^{j}_l,d)-
1/2)} (1 - \hat{p}_t^{j})^{ \alpha_t(A^{j}_r,d) - 1/2)}  } \\
& \cdot \sqrt{\dfrac{n_1(A^{j}_l) n_2(A^{j}_l)}{n(A)}} \cdot \dfrac{
(\hat{p}^j)^{n(A^{j}_l)}
(1- \hat{p}^j)^{n(A^{j}_r)} }{ \prod_t (\hat{p}_t^{j})^{n_t(A^{j}_l)} (1-
\hat{p}_t^{j})^{n_t(A^{j}_r)} } \\
& \simeq C \sqrt{n(A)} \dfrac{ (\hat{p}^j)^{n(A^{j}_l)}
(1- \hat{p}^j)^{n(A^{j}_r)} }{ \prod_t (\hat{p}_t^{j})^{n_t(A^{j}_l)} (1-
\hat{p}_t^{j})^{n_t(A^{j}_r)} },
\end{align*}
 for some $C>0$. 
Then,
$
 \log\Lambda_{m,d}(A,j,\vect{x}_1, \vect{x}_2)  \propto \frac{1}{2}\log n(A) +
\log\Lambda,
$
where 
$$
\Lambda = 
(\hat{p}^j)^{n(A^{j}_l)}
(1- \hat{p}^j)^{n(A^{j}_r)}/\prod_t (\hat{p}_t^{j})^{n_t(A^{j}_l)} (1-
\hat{p}_t^{j})^{n_t(A^{j}_r)}.
$$
Then $\Lambda$
 is the likelihood ratio for testing composite
hypotheses $H_0: p_1^j = p_2^j = p^j$ vs $H_1: p_1^j,  p_2^j \in (0,1)$.
Under the null, $-2 \log \Lambda \overset{d}{\to} \chi_1^2$ by
\cite{wilks_1938} and so $\log\Lambda_{m,d}(A,j,\vect{x}_1, \vect{x}_2)/ \log n 
\overset{p}{\to} \eta'$ for some $\eta'>0$.

Additionally,  $\log \Lambda$ can be also written as
$
\log \Lambda = n(A)Y_{n(A)},
$
$$ Y_{n(A)} \overset{p}{\to}   \big{(} G (p^j) -  \beta 
G(p_1^{j}) - (1-\beta )  G(p_2^{j}) \big{)},
$$
where
$$
\begin{array}{ll}
G(x) = x \log(x) + (1-x) \log(1-x) & \text{ for } x \in (0,1).
\end{array}
$$
If $p_1^j\neq p_2^j$, then $
 G (p^j) -  \beta 
G(p_1^{j}) - (1-\beta )  G(p_2^{j})  < 0,
$
since $G$ is convex and $p^j = \beta p_1^j + (1-\beta)p_2^j $. Thus, 
$\log\Lambda_{m,d}(A,j,\vect{x}_1, \vect{x}_2)/ n \overset{p}{\to} \eta'' $ for
some $\eta'' < 0$.
\end{proof}

\begin{mylemma}\label{lemma:rho_gd}
Consider $A \in
\mathcal{A}^{(\infty)}$  such that 
\begin{enumerate}
 \item  $p_1^j = p_2^j$ for all $j=1, \ldots,
p$.
 \item $\rho_{g,h}(A) \in (0,1)$ for all $g,h\in\{d,m\}$.
 \item $\rho_{g,h}(A_i^j) \in (0,1)$ for all  $g\in\{d,m\}$
and $h\in\mathcal{G},  i=l,r, j=1, \ldots p$.
\end{enumerate}

Then, for all $g \in \mathcal{G}$,
 $$
\rho_{g,d}(A|\vect{x}_1,\vect{x}_2 ) \overset{p}{\to} 0.
$$
\end{mylemma}
\begin{proof}
For  $g = s$, note that $\rho_{s,d}(A|\vect{x}_1,\vect{x}_2) =
0 $, since $\rho_{s,d}(A)=0$ by design.
For  $g \in \{m,d\}$,
$$
\rho_{g,d}(A|\vect{x}_1,\vect{x}_2) = \dfrac{\rho_{g,d}(A)
Z(A,d,\vect{x}_1,\vect{x}_2) }{ \sum_{h \in \mathcal{G}} \rho_{g,h}(A)
Z(A,h,\vect{x}_1,\vect{x}_2)} \leq \dfrac{\rho_{g,d}(A)
Z(A,d,\vect{x}_1,\vect{x}_2) }{\rho_{g,m}(A)
Z(A,m,\vect{x}_1,\vect{x}_2)}
$$
Since $\rho_{g,d}(A)/\rho_{g,m}(A)$
is bounded for Condition (2), it is sufficient to  
that 
$$
\dfrac{ Z(A,d,\vect{x}_1,\vect{x}_2) }{
Z(A,m,\vect{x}_1,\vect{x}_2)} = \dfrac{ \sum_j Z_j(A,d,\vect{x}_1,\vect{x}_2) }{
\sum_j Z_j(A,m,\vect{x}_1,\vect{x}_2)  }  \overset{p}{\to} 0.
$$
For any $j=1, \ldots, p$
$$
\dfrac{  Z_j(A,d,\vect{x}_1,\vect{x}_2) }{
 Z_j(A,m,\vect{x}_1,\vect{x}_2)  } =
\Lambda_{m,d}(A,j,\vect{x}_1, \vect{x}_2)^{-1} \dfrac{ \lambda_j(A,d)\prod_i
\Phi(A^j_i,d,\vect{x}_1,\vect{x}_2)}{\lambda_j(A,m)\prod_i\Phi(A^j_i,m,\vect{x}
_1,\vect{x}_2) },
$$
where
$\Lambda_{m,d}(A,j,\vect{x}_1, \vect{x}_2)^{-1}
\overset{p}{\to} 0 $ by Lemma \ref{lemma:lambda_md},  $\lambda_j(A,d) /
\lambda_j(A,m)$ is bounded by conditions of Theorem 1, and , for  
$g^* = \text{argmax}_{g \in \mathcal{G}}
Z(A^j_i,g,\vect{x}_1,\vect{x}_2)$,  
\begin{align*}
\dfrac{\Phi(A^j_i,d,\vect{x}_1,\vect{x}_2) }{
\Phi(A^j_i,m,\vect{x}_1,\vect{x}_2) } & = 
\dfrac{ \sum_{g \in \mathcal{G}} \rho_{d,g}(A_i^j)
Z(A^j_i,g,\vect{x}_1,\vect{x}_2) }{
\sum_{h \in \mathcal{G}} \rho_{m,h}(A_i^j) Z(A^j_i,h,\vect{x}_1,\vect{x}_2) }\\
& \leq \dfrac{ 
Z(A^j_i,g^*,\vect{x}_1,\vect{x}_2) }{
\sum_{h \in \mathcal{G}} \rho_{m,h}(A_i^j)Z(A^j_i,h,\vect{x}_1,\vect{x}_2) }
\leq \dfrac{1}{\rho_{m,g^*}(A_i^j)}
\end{align*}
is bounded since $\rho_{g,h}(A_i^j) \in (0,1)$ for all $g \in\{
d,m\}, h\in\mathcal{G}$, $i=l,r$ and $j=1, \ldots, p$.
For any $\epsilon>0$, define
$$
E_j(\epsilon) =  \bigg{\{} \omega :
\dfrac{Z_j(A,d,\vect{x}_1,\vect{x}_2)}{Z_j(A,m,\vect{x}_1,\vect{x}_2)}  <
\dfrac{\epsilon}{p} \bigg{\}},
$$
then on $\cap_j E_j(\epsilon)$, 
$$
\dfrac{ \sum_j Z_j(A,d,\vect{x}_1,\vect{x}_2) }{
\sum_j Z_j(A,m,\vect{x}_1,\vect{x}_2)  }  < \epsilon.
$$
Since $Z_j(A,d,\vect{x}_1,\vect{x}_2)/Z_j(A,m,\vect{x}_1,\vect{x}_2)
\overset{p}{\to} 0 $, then $\forall \delta >0, \exists N_j(\epsilon, \delta)$
such that
$\Pr(E_j^c(\epsilon))< \delta/ p$, for
any $n_1,n_2>N_j(\epsilon, \delta)$. Then, for $N(\epsilon, \delta) =
\text{max}_j N_j (\epsilon, \delta)$ 
$$ \Pr(\cap_j E_{j}(\epsilon) )  = 1 - \Pr ( \cup_j E_j^c(\epsilon) ) \geq 1 -
\sum_j
\Pr(E_j^c(\epsilon))  \geq 1 - p \hat{\epsilon} = 1 - \delta.
$$
and this implies that $\forall \epsilon > 0$ and $\forall \delta >0 $, $\exists$
$N(\epsilon,\delta)$ such that, for
any $n_1,n_2>N(\epsilon,\delta)$,
$$
\Pr \bigg{(}  \dfrac{Z(A,d,\vect{x}_1,\vect{x}_2)}{Z(A,m,\vect{x}_1,\vect{x}_2)}
   <
\epsilon  \bigg{)}
\geq 1 - \delta.
$$
\end{proof}

\begin{mytheorem}
\emph{(Consistency under the alternative)}
\label{th:altcase}

\end{mytheorem}
\begin{proof}
First we find the tree topologies with 
highest marginal posterior probability under the null component of $\vect{Q}$.
To this
end, assume $ (\vect{x}_1 \cup \vect{x}_2) \sim \tilde{Q}$ and $\tilde{Q} \sim
\text{OPT}(\tilde{\vect{\rho}}, \tilde{\vect{\lambda}},
\tilde{\vect{\alpha}}, Q_0)$, where   $\tilde{\rho}_0(A) =
\rho_{m,s}(A)/(\rho_{m,m}(A)+\rho_{m,s}(A))$,  
$\tilde{\lambda}_j(A)  = \lambda_j(A,m)$, and
  $\tilde{\vect{\alpha}}^j(A)  = \vect{\alpha}^{j}(A,m)$
for all $A \in \mathcal{A}^{(\infty)}$ and $j=1, \ldots, p$.
Define $\pi^{(k)}$ a tree topology where all regions are stopped at most at
level $k$, and call $\Pi^{(k)}$ the set of the $\pi^{(k)}$'s  with highest
marginal posterior probability as $n
\rightarrow
\infty$.
Since there exists $A\in \mathcal{A}^{(\infty)}$ where $P_1(A_i^j|A) \neq
P_2(A_i^j|A)$ and $\beta P_1(A_i^j|A) + (1-\beta) P_2(A_i^j|A) \neq
Q_0(A_i^j|A)$, by Theorem~4 in \cite{wong_ma_2010}   for some large enough
$k$, there is at least one non-stopped
region $B$
arising in $\pi^{(k)}$ 
 such that $P_1(B_i^j|B) \neq
P_2(B_i^j|B)$, for
each $\pi^{(k)} \in \Pi^{(k)}$.
 Then, we show that  the likelihood of ``dividing'' on $B$ dominates the 
likelihood of ``merging''.
In fact,
\begin{align*}
\dfrac{Z_j(B, m, \vect{x}_1 , \vect{x}_2)}{Z_j(B, d,
\vect{x}_1, \vect{x}_2)} & = 
\dfrac{\lambda_j(B,m)}{\lambda_{j}(B,d)}\Lambda_{m,d}(B,j,\vect{x}_1,
\vect{x}_2) \prod_i 
\dfrac{\Phi(B_i^{j}, m,
\vect{x}_1 ,\vect{x}_2)}{ \Phi(B_i^j, d, \vect{x}_1, \vect{x}_2) },
\end{align*}
where $\log \Lambda_{m,d}(B,j,\vect{x}_1, \vect{x}_2)\overset{p}{\to} - \infty$ 
 by
Lemma \ref{lemma:lambda_md}, $\lambda_{j}(B,d)/\lambda_{j}(B,m) $ is bounded by
Theorem~\ref{thm:large_support}, and 
$$
\dfrac{\Phi(B_i^{j}, m,
\vect{x}_1 ,\vect{x}_2)}{ \Phi(B_i^j, d, \vect{x}_1, \vect{x}_2) } \leq
\max_{g\in \mathcal{G}} \dfrac{\rho_{m,g}(B_i^{j})}{\rho_{d,g}(B_i^{j})}
$$
is bounded since $\rho_{d,g}(B_i^j) \in (0,1)$ for all
$g \in
\mathcal{G}$. This implies that 
$$
\Pr( H_0| \vect{x}_1, \vect{x}_2) \overset{p}{\to} 0 \text{ under }
P_1^{( \infty)} \times P_2^{( \infty)}. 
$$
\end{proof}

\begin{mytheorem}
\emph{(Consistency under the null)}
\label{th:null_case}

\end{mytheorem}
\begin{proof}
For any set $A \in \mathcal{A}^{l}$ for $l>k$,  we have that
$\Psi(A,g,\vect{x}_1,\vect{x}_2) = 1,$
since $\rho_{g,d}(A) = 0$ by design. 
For any set $ A \in \mathcal{A}^{(k)}$, if
$\Psi(A^j_i,g,\vect{x}_1,\vect{x}_2)
\overset{p}{\to} 1$ for any $j=1, \ldots, p$, and $i \in
\{l,r\}$,
then, by Slutsky's theorem,
$$
\Psi(A,g,\vect{x}_1,\vect{x}_2)  \overset{p}{\to}
\rho_{g,s}(A|\vect{x}_1,\vect{x}_2) +
\rho_{g,m}(A|\vect{x}_1,\vect{x}_2)
\sum_j \lambda_j(A,m|\vect{x}_1,\vect{x}_2), 
$$
for any $g \in \{d,m\}$.
Additionally, $\sum_j \lambda_j(A,m|\vect{x}_1,\vect{x}_2) = 1$, and 
 $ \big{(} \rho_{g,s}(A|\vect{x}_1,\vect{x}_2) +
\rho_{g,m}(A|\vect{x}_1,\vect{x}_2) 
 \big{)} \overset{p}{\to} 1 $ by Condition (1), Condition (2) and Lemma
\ref{lemma:rho_gd}.
Then,
$$
\begin{array}{ccc}
\Psi(A,g,\vect{x}_1,\vect{x}_2) \overset{p}{\to} 1 & \text{and} &
\Pr( H_0| \vect{x}_1, \vect{x}_2)
\overset{p}{\to} 1. 
\end{array}
$$
\end{proof}

\subsection{Example 1 in numerical examples}

 In Figure \ref{fig:example1c}
 we provide an analysis of sensitivity to the prior choice for the local shift
difference example. We plot 
$\Pr(H_0|\vect{x}_1,\vect{x}_2)$
  for different choices of
the hyperparameters of the transition probability matrix under the null (red)
and under the alternative (blue), showing that the
results are robust to the prior specification. Similar results are obtained for
the other scenarios and examples and so are not reported. 

\begin{figure}
 \centering
\makebox{\includegraphics[width = 1\textwidth]{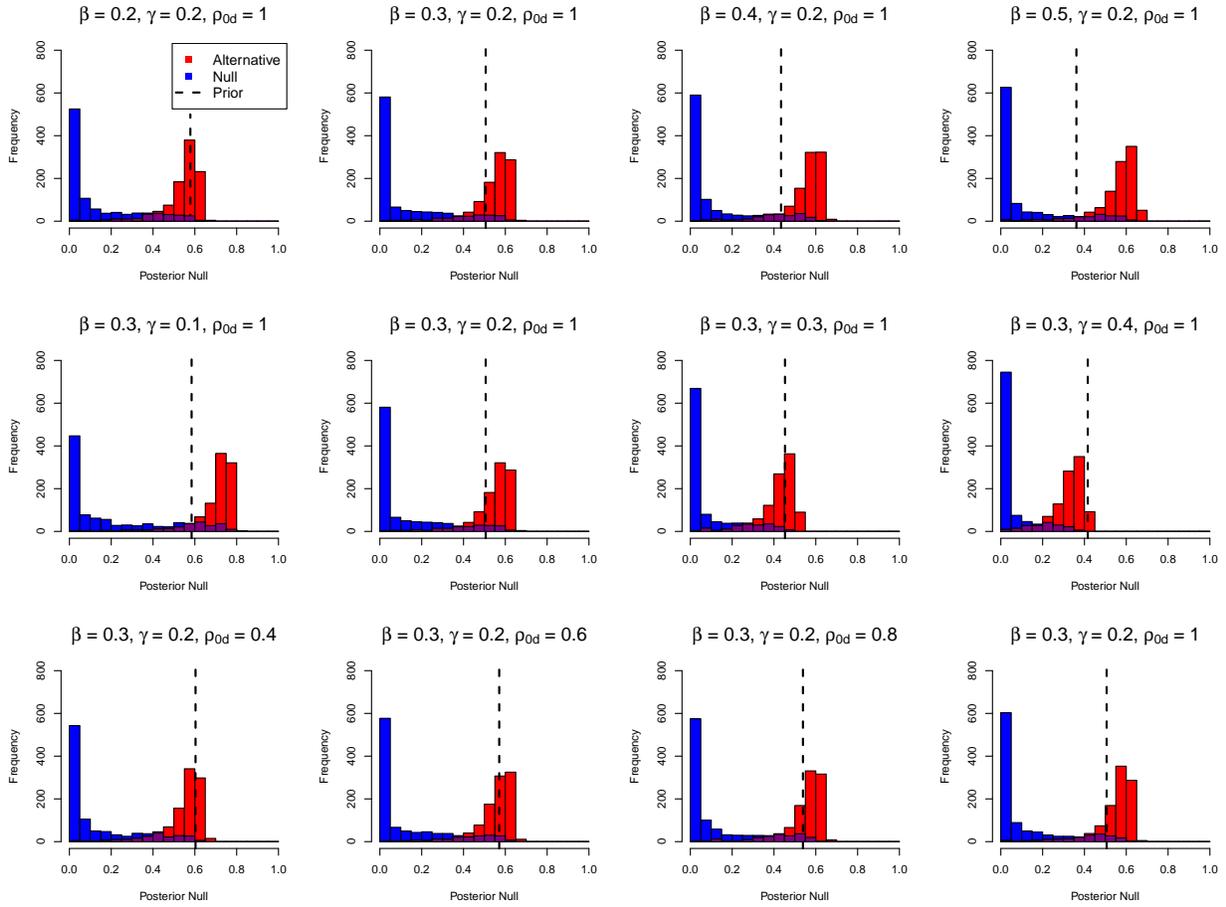}}
\caption{Bayes Factors for the local shift difference scenario
in $\mathbb{R}^1$ for different choices of the hyperparameters of the
transition probability matrix. In red the results under the
null and in blue the results under the alternative. }
\label{fig:example1c}
\end{figure}

\bibliographystyle{Chicago}
\bibliography{citations_database}{}

\end{document}